\UseRawInputEncoding

\documentclass[12pt]{amsart}
\usepackage{amsmath}
\usepackage{amssymb}
 \usepackage{CJK}
\usepackage[all]{xy}

\newcommand{\corr}[1]{\langle {#1} \rangle}

\newcommand{\half}{\frac{1}{2}} \newcommand{\vac}{|0\rangle}

\newcommand{\bp}{{\mathbf p}} \newcommand{\bs}{\mathbf{s}}
\newcommand{\bt}{{\bf t}}  \newcommand{\bT}{\mathbf{T}}
\newcommand{\bX}{\mathbf{X}} 
  
 \newcommand{\cA}{{\mathcal A}}
\newcommand{\cB}{{\mathcal B}} \newcommand{\cH}{{\mathcal H}}
\newcommand{\cC}{{\mathcal C}} \newcommand{\cL}{{\mathcal L}}
 \newcommand{\cP}{{\mathcal P}} \newcommand{\cQ}{{\mathcal Q}}
 \newcommand{\cS}{{\mathcal S}}
  \newcommand{\cZ}{{\mathcal Z}}
 \newcommand{\bR}{{\mathbb R}}
\newcommand{\bC}{{\mathbb C}}

 \newcommand{\fh}{{\mathfrak h}}

\newcommand{\pd}{{\partial}}

\DeclareMathOperator{\ev}{ev}

\newcommand{\be}{\begin{equation}}
\newcommand{\ee}{\end{equation}}
\newcommand{\ben}{\begin{eqnarray*}}
\newcommand{\een}{\end{eqnarray*}}
\newcommand{\bea}{\begin{eqnarray}}
\newcommand{\eea}{\end{eqnarray}}

\newcommand{\Brac}[2]{\begin{Bmatrix} {#1} \\ {#2} \end{Bmatrix}}

\usepackage{color}

\definecolor{yellow}{rgb}{1,1,0}
\definecolor{orange}{rgb}{1,.7,0}
\definecolor{red}{rgb}{1,0,0}
\definecolor{white}{rgb}{1,1,1}

\definecolor{A}{rgb}{.75,1,.75}

\DeclareMathOperator{\Li}{Li}

\newtheorem{Theorem}{Theorem}[section]
\newtheorem{Theorem/Definition}{Theorem/Definition}[section]
\newtheorem{Proposition}{Proposition}[section]

\newtheorem{Corollary}{Corollary}[section]

\begin{document}
\title[Mathematics Related to the Interpolating Statistics]
{On Some Mathematics Related to the Interpolating Statistics}

\author{Jian Zhou}
\address{Department of Mathematical Sciences\\Tsinghua University\\Beijng, 100084, China}
\email{jzhou@math.tsinghua.edu.cn}

\begin{abstract}
Motivated by fractional quantum Hall effects, we introduce a universal space of statistics
interpolating Bose-Einstein statistics and Fermi-Dirac statistics.
We connect the interpolating statistics to umbral calculus
and use it as a bridge to study the interpolation statistics by
the principle maximum entropy by deformed entropy functions.
On the one hand this connection makes it possible to relate fractional quantum
Hall effects to many different mathematical objects, including formal group
laws, complex bordism theory, complex genera, operads, counting trees, spectral curves
in Eynard-Orantin topological recursions, etc.
On the other hand,
this also suggests to reexamine umbral calculus from the point of view of quantum mechanics
and statistical mechanics.
\end{abstract}


\maketitle


\section{Introduction}

Because statistical mechanics generally involves counting numbers of states,
it is naturally expected that it is closely related to enumerative combinatorics.
Traditionally,
three different rules of counting states
give rise to Boltzmann-Gibbs statistics, Bose-Einstein statistics, and Fermi-Dirac statistics£¬
respectively.
The number of states of $n_i$ identical particles of type $i$ occupying a group of $g_i$ states is, respectively, given by
\ben
W_{BE}(\{n_i\}) & = & \prod_i \frac{(g_i + n_i -1)!}{n_i! (g_i -1)!},  \\
W_{FD}(\{n_i\}) & = & \frac{g_i!}{n_i! (g_i-n_i)!}, \\
W_{BG}(\{n_i\}) & = & \frac{g_i^{n_i}}{n_i!}.
\een
(See e.g. \cite[(8.40)-(8.42)]{Huang}.)
However£¬
 these counting problems   can be solved by elementary methods,
hence  more advanced combinatorial theories have not been applied in their study.

In 1982, the fractional quantum Hall effect \cite{Tsui-Stormer-Gossard} was discovered
by Tsui, Stormer and Gossard.
This leads to many new  progresses in statistical physics.
Following the original work of Laughlin \cite{Laughlin} describing the ground state of the
FQHE for a filling fraction $\nu = \frac{1}{m}$ (where $m$ is odd),
Halperin \cite{Hal} conjectured and Arovas, Shrieffer and Wilczek \cite{Aro-Sch-Wil}
showed that the quasiparticles and quasiholes in quantum Hall systems not only have
fractional charges, they also obey fractional statistics in the sense that
they are anyons in the sense of Wilczek \cite{Wil} (see also Leinnaas-Myrheim \cite{Lei-Myr}).
Suppose that we have two identical particles in two dimensions.
Then when one particle is exchanged in a counterclockwise manner with the other, the wavefunction can
change by an arbitrary phases
\be
\psi(\vec{r}_1,\vec{r}_2) \to e^{i\theta}\psi (\vec{r}_1,\vec{r}_2),
\ee
The special cases $\theta = 0, \pi$ correspond to bosons and fermions,
respectively.
Particles with other values of the angle $\theta = \alpha \pi$ are called anyons.
Fractional statistics in this sense involves geometric phases \cite{Berry},
and it indicates the topological nature of the fractional quantum Hall system.
In fact, fractional quantum Hall effect is an example of topological order,
a notion introduced
by Wen \cite{Wen} in 1989.

In 1991, Haldane \cite{Hal} reformulated the  concept of ``fractional statistics"
as a generalization of the Pauli exclusion principle.
He considered the situations in which
the number $g_i$ of independent states of a single particle  of type $i$ is finite
and extensive. This means the quasiparticles and qausiholes are  confined to a finite
region, and their number is proportional
to the size of region.
The generalized Pauli exclusion principle proposed by Haldane is the following linear relation:
\be
\Delta g_i = - \sum_j g_{ij} \Delta n_j.
\ee
Let $d_i$ be the value of $g_i$ when $n_i =1$,
then the number of $\{g_i\}$ for general $\{n_i\}$ are
\be
g_i = d_i + n_i-1 - \sum_j g_{ij}(n_j -\delta_{ij}).
\ee
For state-counting purposes at fixed particle numbers
the particles can be regarded as bosons with a Fock space of dimension $g_i$,
or fermions with a Fock space of dimension $g_i+n_i-1$,
so the number of states for given $\{n_i\}$ is given by
the following formula in Wu \cite[(4)]{Wu}:
\be
W(\{n_i\})
= \prod_i \frac{(g_i+(n_i-1)-\sum_j g_{ij}(n_j-\delta_{ij})]!}{n_i!
(g_i-1-\sum_j g_{ij}(n_j-\delta_{ij})]!}.
\ee
When $g_{ij}=\beta \delta_{ij}$,
it becomes
\be
W(\{n_i\})
= \prod_i \frac{(g_i+(n_i-1)(1-\beta)]!}{n_i!(g_i-1-\beta(n_i-1)]!}.
\ee
According to the fundamental principles of quantum statistical mechanics,
the grand partition function is given by
\be
Z = \sum_{\{n_i\}} W(\{n_i\})\exp \biggl[\sum_i (\mu_i-\epsilon_i)/kT \biggr],
\ee
where $T$ is the temperature, $k$ is the Boltzmann constant,
$\epsilon_i$ is the energy level of a particle of species $i$,
and $\mu_i$ is the chemical potential for species $i$.
Then one has
\be
Z = \prod_i Z_i,
\ee
where $Z_i$  is given by
\be
Z_i = \sum_{n_i} \frac{(g_i+(n_i-1)(1-\beta)]!}{n_i!(g_i-1-\beta(n_i-1)]!}
\exp \biggl[(\mu_i-\epsilon_i)/kT \biggr].
\ee
This leads to the notion of statistics interpolating between Bose and Fermi statistics
\cite{Cha-Sri}.
We will refer to them as interpolating statistics.
For a comparative study of various interpolating statistics,
see \cite{Cha-Sri}.

As our point of departure,
we will  consider the universal space $\cS$ of all possible interpolating statistics
and study some mathematics related to it.
This is consistent with a common mathematical practice: Form a set by collecting all mathematical
objects with the same properties, and then study the mathematical structures on this set.
This is also inspired by Wilson's renormalization theory in which he considered
the space of all coupling constants and the dynamics on it.
The basic mathematical objects in our study are formal power series,
so ultimately  the mathematical setting should be that of formal algebraic geometry.
However,
on the outset we will first establish a link between interpolating statistics to umbral calculus in
combinatorics,
and we will use this as a stepping stone to applications from the point of formal algebraic geometry
in subsequent work.

Before 1970s, umbral calculus refered to some powerful techniques to prove combinatorial identities
introduced by Blissard in 1861.
In the 1970s, Rota and his collaborators developed the mathematical
foundations of umbral calculus by means of linear functionals on the space of polynomials
\cite{Rota et al, Roman-Rota, Roman}.
Another approach to this theory is through the techniques of
Hopf algebras \cite{Joni-Rota}.
Clearly a connection between umbral calculus and quantum mechanics was known to Rota and his collaborators,
but this has not been fully elaborated.
The link with interpolating statistics provides more stimulus to interpret umbral calculus
from the point of view of quantum mechanics.
In the text we will reexamine umbral calculus using representations of Heisenberg commutation relations,
coherent states, Segal-Bargmann transform, Wick quantization, etc.

An application of the link between the interpolating statistics and the umbral calculus
is the study of interpolating statistics by the principle of maximum entropy.
This principle was first expounded by Jaynes in two papers in 1957 \cite{Jaynes1, Jaynes2}
where he emphasized a natural correspondence between statistical mechanics and information theory.
According to this principle,
Boltzmann distribution can be derived by maximizing the Shannon entropy  under
the constraint of a fixed expectation value of energy.
To derive the Bose-Eisnstein distribution and Fermi-Dirac distribution,
or more generally, all the interpolating statistics,
one needs to find suitable generalized entropy functions to be maximized.
In 1967, Havrda-Charv\'at introduced
structural $\alpha$-entropy \cite{Hav-Cha}  within information theory.
In 1988, Tsallis \cite{Tsallis88} rediscovered this entropy and proposed to use it
to study nonextensive statistical mechanics.
For more information,
see his monograph \cite{Tsallis09} published in 2009.
The Tsallis $q$-entropy is defined by
\be
S = \frac{1- \sum_i p_i^q}{q-1}
= -\sum_i p_i^q \cdot \frac{p_i^{1-q}-1}{1-q}.
\ee
In 1994, Tsallis \cite{Tsallis94} proposed to use the $q$-logarithm function
\be
\log_q x:= \frac{x^{1-q}-1}{1-q}
\ee
and the $q$-exponential function
\be
\exp_q(x):= [1+(1-q)x]^{\frac{1}{1-q}}
\ee
to unify the expression in traditional Boltzmann-Gibbs statistical mechanics
and his nonextensive statistical mechanics.
This idea was generalized by Naudts \cite{Naudts02} in 2002.
He introduced more general deformed exponentials and logarithms in generalized statistical physics.
See also his monograph \cite{Naudts11} published in 2011.

We make a connection between umbral calculus and the generalized entropy based on the following observation:
In umbral calculus,
one deals with the deformed power functions,
and so their exponential generating functions can be regarded as deformed exponential functions.
As a consequence of the connections of umbral calculus to both the interpolating statistics
and the generalized entropy,
we establish a connection between interpolating statistics and generalized entropy.
Our main result is that the generalized entropy function is essentially
the Legendre transform of the free energy (cf. Theorem \ref{thm:Main}),
as in ordinary thermodynamics.
This demonstrates the surprising power of applying umbral calculus to study
various aspects of interpolating statistics.
Such applications will be further pursued in subsequent work.

The rest of this paper is arranged as follows.
In Section \ref{sec:Statistics} we recall some definitions related to the statistics
interpolating Bose-Einstein statistics and Fermi-Dirac statistics.
In a very short Section \ref{sec:Space} we introduce the space
of interpolating statistics and a sequence of group structures on it.
We also introduce two different but related ways to associated spectral curves
to interpolating statistics.
We relate interpolating statistics to umbral calculus in Section \ref{sec:UC}.
We give umbral interpretation of the group structure of the space
$\cS$ of interpolating statistics.
Next  we relate umbral calculus to quantum mechanics in Section \ref{sec:QM}.
In particular,
generating series of associated series in umbral calculus are interpreted as
coherent states in quantum mechanics.
In Section \ref{sec:Deformed-Entropy} we interpret such generating series in
umbral calculus as deformed exponential functions
and this motivates us to relate interpolating statistics to generalized entropy functions.
Some concluding remarks are made in Section \ref{sec:Conclusions}.
We present many examples in Appendix A. Most of them are from Roman \cite[Chapter 4]{Roman},
but we present some computations related to their relationship with generalized entropy.
Some new examples are also presented there.
Many sequences on The On-Line Encyclopedia of Integer Sequences \cite{Sloane}
appear in these examples,
indicating very rich combinatorial interpretations to be explored.

\section{Statistics Interpolating Bose and Fermi Statistics}

\label{sec:Statistics}

In this Section we recall the general framework
for interpolations between Bose-Einstein and Fermi-Dirac statistics
at the level of grand canonical partition functions as developed by
 Chaturvedi and Srinivasan \cite{Cha-Sri}.
We will also define the entropy associated with the one-particle free enrrgy
by Legendre transformation.

\subsection{Grand canonical partition functions of noninteracting particles}

Assume that for a system of $N$ identical particles each of which can occupy
$n$ states corresponding to energies $E_1 <  \dots < E_n$,
the canonical partition function has the following structure:
\be
Z_N(x) = \sum_{\sum_{i=1}^n N_i = N} W_{N_1} \cdots W_{N_n} x_1^{N_1} \cdots x_n^{N_n},
\ee
where $x=x_1, \dots, x_n$; $x_i = \exp(-\beta E_i)$.
This assumption means that there is no interactions between particles at different energy levels,
and there are $W_{N_i}$ particles at energy level $E_i$ for each $i=1, \dots, n$.
Here $W_0 = 1$ is always assumed.
In most of our example,
we will also assume that $W_1 =1$.

Under the above assumption,
the grand partition canonical function has the following factorized form:
\be \label{eqn:GCP}
\cZ(X_1, \dots, X_n) = \sum_{N \geq 0}  e^{N\mu \beta} Z_N(x)
= \prod_{i=1}^n z (X_i),
\ee
where $X_i = \exp(-\beta(E_i-\mu))$ and
\be
z(X) = \sum_{n=0}^\infty W_n X^n
\ee
is called the {\em one-particle partition function}.

\subsection{Examples}

The Bose-Einstein, Maxwell-Boltzmann and Fermi-Dirac statistics are given by
\begin{align}
z^{BE}(X) & = \frac{1}{1 - X}, &
z^{FD}(X)& =1 + X, & z^{BG}(X) & = e^X,
\end{align}
respectively.

\subsection{Occupation numbers}

In the following expansion:
\be
z(X)^N = \sum_{k \geq 0} W_k(N) X^k,
\ee
the coefficients $W_k(N)$ are the numbers of $N$ particles occupying $k$ states,
and they are called the occupation numbers.

From the fact that $z(X)^{N_1} \cdot z(X)^{N_2} = z(X)^{N_1+N_2}$,
one gets the following recursion relations for occupation numbers:
\be \label{eqn:Occ}
W_k(N_1+N_2) = \sum_{i=0}^k W_i(N_1) W_{k-i}(N_2).
\ee

\subsection{Mean number of particles}

The expression for the mean number of particles in the
energy state $E_i$ for this class of statistics reads
\be
w(X_i) = X_i \frac{\pd}{\pd X_i} \log \cZ(W) = X_i \frac{\pd}{\pd X_i} \log z(X_i).
\ee
It will be called the {\em weight function} in this paper.
For Bose-Einstein, Maxwell-Boltzmann and Fermi-Dirac statistics we have
\begin{align}
w^{BE}(X) & = \frac{X}{1 - X}, &
w^{FD}(X) & = \frac{X}{1 + X}, &
w^{BG}(X) & = X,
\end{align}
respectively.
When $X = e^{-\beta(E-\mu)}$,
one gets the familiar expressions:
\begin{align}
w^{BE} & = \frac{1}{e^{\beta(E-\mu)} - 1}, &
w^{FD} & = \frac{1}{e^{\beta(E-\mu)} + 1}, &
w^{BG} & = e^{-\beta(E-\mu)}.
\end{align}

\subsection{Cluster coefficients}

We now expand the one-particle free energy as follows:
\be
F(X):=\log z(X) = \sum_{n=1}^\infty \frac{w_n}{n} X^n,
\ee
and so
\be
w(X) = \sum_{n=1}^\infty w_n X^n.
\ee

The coefficients $w_n$'s are called the {\em cluster coefficients}.
One can easily find the following combinatorial formula expressing the $W_n$'s as
polynomials in the  $w_n$'s:
\be
W_n = \sum_{\sum m_ii = n} \prod_i \frac{w_i^{m_i}}{i^{m_i}m_i!}.
\ee
For example,
\ben
&&W_1 = w_1, \\
&&W_2 = \frac{w_2}{2}+\frac{w_1^2}{2}, \\
&&W_3 = \frac{w_3}{3} + \frac{w_2w_1}{2} + \frac{w_1^3}{3!}.
\een
Conversely,
the sequence $\{w_n\}_{n \geq 1}$ is obtained from the sequence $\{W_n\}_{n \geq 1}$by:
\be
\sum_{n=1}^\infty \frac{w_n}{n} X^n
= \log (1+\sum_{n=1}^\infty W_n X^n).
\ee
For example,
\ben
&&w_1 = W_1, \\
&&w_2 = W_2-\frac{W_1^2}{2}, \\
&&w_3 = W_3 - W_2W_1 + \frac{W_1^3}{3}.
\een

\subsection{Entropy associated with the one-particle partition function}
\label{sec:Entropy}

The entropy function associated to the one-particle partition function $z=e^{F(x)}$
is obtained by the Legendre transformation of the free energy $F(X)$ as follows:
\be
H(X) := F(X) - \log X \cdot \frac{d F(X)}{d \log X}  =F(X) - \log X \cdot w(X).
\ee
The entropy is also called the effective action or one-particle irreducible partition correlation function.
As usual,
if one can invert the function $w=w(X)$ to get $X = X(w)$,
then $H$ can be expressed in terms of $w$:
\be
H = - w \log X(w) + F(X(w)).
\ee
For example,
we have
\ben
H^{BG} & = & - X \log X + X = -w^{BG} \log w^{BG} +w^{BG}, \\
H^{BE} & = & - \frac{X}{1-X} \cdot \log X + \log \frac{1}{1-X} \\
& = & -w^{BE} \log \frac{w^{BE}}{1+w^{BE}} + \log (1+w^{BE}), \\
H^{FD} & = & -\frac{X}{1+X} \cdot \log X + \log (1+X) \\
& = & - w^{FD} \cdot \log \frac{w^{FD}}{1-w^{FD}} + \log \frac{1}{1-w^{FD}}.
\een

\section{Space of Interpolating Statistics and Its Group Structures}

\label{sec:Space}

As our point of departure,
we will introduce in this Section the space $\cS$ of interpolating statistics,
on it an involution that interchanges the Bose-Einstein statistics
with the Fermi-Dirac statistics,
and a sequence of group structures on $\cS$.
Furthermore,
we define two versions of spectral curves associated with
elements in $\cS$.

\subsection{Space of interpolating statistics}

Since we have $w_1 = W_1 = 1$,
$w(X)$ is  of the form
\be
w(X) = X+ \sum_{n \geq 2} w_n X^n.
\ee
Given any series of this form,
one gets a statistics
whose grand canonical partition function is given by \eqref{eqn:GCP}.
So we now introduce a {\em space of interpolating statistics}:
\be
\cS = \{ w(X) = X + \sum_{n \geq 2} w_n X^n\}.
\ee

\subsection{Dual cluster coefficients and dual statistics}
\label{sec:Dual-Stat}

Given any formal sequence $w = X + \sum_{n \geq 2} w_n X^n \in \cS$,
by Lagrangian inversion,
one can show that
\be \label{eqn:X-in-w}
X = w + \sum_{n \geq 2} \hat{w}_n w^n
\ee
for some weighted homogeneous polynomials $\hat{w}_n$ in $\{w_2, w_2, w_n\}$ of degree $n-1$,
where $\deg w_j = j-1$.
The coefficients $\hat{w}_n$ will be called the {\em dual cluster coefficients}.
Therefore,
one can define the dual sequence of $w$ by:
\be
\hat{w}(X) = X + \sum_{n \geq 2} \hat{w}_n X^n.
\ee
The map $w \mapsto \hat{w}$ defines an involution $\sigma$ on $\cS$.
It is clear that $w_{BG}(X) = X$ is the only fixed point of this involution.
One can easily check that
\begin{align}
\hat{w}_{BE}(X) & = \frac{X}{1+X} = w_{FD}(X), &
\hat{w}_{FD}(X) & = \frac{X}{1-X} = w_{BE}(X).
\end{align}
In other words,
the involution $\sigma: \cS \to \cS$ interchanges the Bose-Einstein statistics with the Fermi-Dirac statistics,
and so we will refer to it as the {\em generalized boson-fermion duality}.

\subsection{A sequence of group structures on $\cS$}
\label{sec:Group1}

The construction in last subsection inspires us to introduce a sequence of group structures
on $\cS$.
Given two elements $w(X) = X + \sum_{n > 1} w_n X^n$ and $v(X) = X + \sum_{n > 1} v_n X^n$
in $\cS$,
one can define $v \circ w$ by composition of series:
\ben
(v\circ w) (X) & = & w(X) + \sum_{n \geq 1} v_n w(X)^n \\
& = & (X + \sum_{k > 1} w_k X^k)
+ \sum_{n > 1} v_n (X + \sum_{k>1} w_k X^k)^n.
\een
Clearly,
$(\cS, \circ)$ is a group.
For $m \geq 1$, and $w = X  + \sum_{N > 1} w_n X^n$,
define
\ben
&& w^{(m)}(X) = X + \sum_{n > 1} n^m w_n X^m,
\een
and define
\ben
v\circ_m w = v^{(m)} \circ w^{(m)}.
\een
We call it the the $m$-th group multiplication on $\cS$.
These group structures will not play a role in this paper.
We expect they are part of bigger symmetry structure on $\cS$
to be discovered in the future.

\subsection{Spectral curves associated with interpolating statistics}

Consider the curve in $(X,z)$-plane defined by
\be
z = z(X)
\ee
or the curve in the $(X, Y)$-plane defined by:
\be
e^Y = z(X).
\ee
They will be called the {\em spectral curves} associated
with the corresponding interpolating statistics.
They are inspired by Eynard-Orantin topological recursion.
See Appendix \ref{sec:Abel} and Appendix \ref{sec:Gould-Special}
for some interesting examples.

\section{Interpolating Statistics and Umbral Calculus}

\label{sec:UC}

In this Section we relate interpolating statistics
to umbral calculus.
This makes powerful combinatorial results accessible to the study of
interpolating statistics.
We will present some examples in Section \ref{sec:Examples}.
We will also give umbral interpretation of the group structure of the space
$\cS$ of interpolating statistics introduced in \S \ref{sec:Group1}.

\subsection{Occupation numbers as deformed binomial coefficients}

Recall the generating series of the occupation numbers
is given by:
\ben
\sum_{k =0}^\infty  W_k(N) X^k
= z(X)^N = e^{NF(X)} = \exp (N \sum_{n=1}^\infty \frac{w_n}{n} X^n).
\een
After expanding the right-hand side of the last equality
as a series in $X$,
one sees that each $W_k(N)$ is a polynomial of degree $k$ in $N$.

One can rewrite $W_k(N)$ as $\Brac{N}{k}$
and understand it as a {\em deformed binomial coefficient}.
This is because the following identities clearly hold for $n \geq 0$
by \eqref{eqn:Occ}:
\be
\sum_{i+j=n} \Brac{x}{i} \Brac{y}{j} = \Brac{x+y}{n}.
\ee
This is understood as the {\em deformed Chu-Vandermonde identities}.

\subsection{Occupation numbers, deformed power functions, and polynomial sequence of binomial type}

Let us write the occupation numbers as follows:
\be
W_k(N) = \frac{1}{k!} \gamma_k(N),
\ee
where the factor $\frac{1}{k!}$ is the Gibbs overcounting correction factor.
Then we have by \eqref{eqn:Occ}:
\be
\gamma_k(N_1+N_2) = \sum_{i=0}^k \binom{k}{i} \gamma_i(N_1) \gamma_{k-i}(N_2).
\ee

Recall a sequence of polynomials $\{\gamma_n(x) \}$
is said to be of {\em binomial type} \cite{Kung-Rota-Yan}
if and only if
\bea
&& \gamma_1(x) =x, \\
&& \gamma_n (x+y)= \sum^n_{j=0} \binom{n}{j} \gamma_j(x) \gamma_{n-j}(y)
\label{eqn:Binomial-Seq}
\eea
for all $x$, $y$  and $n$.
So we are led to the theory of polynomial sequence of binomial type
developed in the setting of umbral calculus \cite{Roman}.

\subsection{Examples of classical umbral calculus}

Umbral calculus is a symbolic method used by mathematicians in the 19th century
to magically ``prove" some identities.
For example,
starting with the binomial formula:
\be \label{eqn:Binomial}
(x+y)^n = \sum_{k=0}^n \binom{n}{k} x^k y^{n-k},
\ee
one gets:
\be
(x+y)_n = \sum_{k=0}^n \binom{n}{k} (x)_k(y)_{n-k},
\ee
where
\be
(x)_n = \frac{\Gamma(x+n)}{\Gamma(x)}.
\ee
Another example involves the Bernoulli polynomials defined by:
\be
B_n(x) = \sum_{k=0}^n \binom{n}{k} B_{n-k} x^k.
\ee
Here is a ``proof" that
\be
\frac{d}{dx} B_n(x) = n \cdot B_{n-1}(x),
\ee
it involves the interchanging of $B^n$ with $B_n$:
\ben
\frac{d}{dx}B_n(x)
&= &\frac{d}{dx} \sum_{k=0}^n \binom{n}{k} B_{n-k} x^k
\sim \frac{d}{dx} \sum_{k=0}^n \binom{n}{k} B^{n-k} x^k  \\
& = & \frac{d}{dx}(B+x)^n = n \cdot (B+x)^{n-1} \\
& = & n \cdot \sum_{k=0}^{n-1}\binom{n-1}{k}B^{n-1-k}x^k \\
& \sim &  n \cdot \sum_{k=0}^{n-1}\binom{n-1}{k}B_{n-1-k}x^k \\
 & = & n \cdot B_{n-1}(x).
\een
Here is one more example \cite[Proposition 4.2.1]{Kung-Rota-Yan}.
Suppose that a sequence $\{b_n\}_{n \geq 0}$is related to $\{a_n\}_{n\geq 0}$ by
the recursion relations:
\ben
&&b_n = \sum_{k=0}^n \binom{n}{k} a_k
\een
for all $n$, then raising the subscripts one gets:
\ben
&&b^n = \sum_{k=0}^n \binom{n}{k} a^k = (a+1)^n,
\een
and so
\ben
&&a^n = (b-1)^n = \sum_{k=0}^n (-1)^{n-k}\binom{n}{k} b^k,
\een
then, lowering the superscripts, one obtains:
\ben
&&a_n = \sum_{k=0}^n (-1)^{n-k}\binom{n}{k} b_k
\een
for all $n$.

Now in \eqref{eqn:Binomial},
changing $(x+y)^n$ to $\gamma_n(x+y)$, $x^k$ to $\gamma_k(x)$, and $y^{n-k}$ to $\gamma_{n-k}(y)$,
one gets \eqref{eqn:Binomial-Seq}.
So umbral calculus is the right setting for studying polynomial sequences of binomial type.

\subsection{Hopf algebra in the modern classical umbral calculus}
Since 1960s,
Rota and his collaborators have constructed a rigorous mathematical
foundation for umbral calculus based on the language of linear functionals,
linear operators and their adjoint operators.
We will present a brief survey in this and the next Subsections.
For  more details,
see Roman \cite{Roman}.

Note the change from $x^n$ to $a_n$ for $n \geq 0$ defines a linear map
\be
L:\bC[x] \to \bC, \quad  x^n \to a_n.
\ee
Denote $L(p(x))$ by $\corr{L|p(x)}$.
Let $D=\frac{d}{dx}$, then one has
\be
D^m x^n =(n)_m x^{n-m},
\ee
also consider the evaluation map:
\be
\ev: \bC[x] \to \bC, \quad p(x) \mapsto p(0).
\ee
Then one has
\be \label{eqn:D-x}
\corr{D^k|x^n} = k!\delta_{k,n},
\ee
and so
\be
\corr{L|p(x)} = \ev(f_L(D)(p(x))),
\ee
where
\be
f_L(t) = \sum_{n=0}^\infty \frac{a_n}{n!} t^n
\ee
is a formal power series.
So one gets a linear isomorphism:
\be
(\bC[x])^* \cong \bC[[t]], \quad L\mapsto f_L(t),
\ee
where  $(\bC[x])^*$ denotes the space of linear functionals on $\bC[x]$.

The space of formal power series is an algebra with multiplication
given by the binomial convolution of the coefficients:
\be
\sum_{k=0}^\infty \frac{a_k}{k!} t^k \cdot \sum_{l=0}^\infty \frac{b_l}{l!} t^l
= \sum_{n=0}^\infty \frac{t^n}{n!} \sum_{k+l=n} \binom{n}{k} a_k b_l.
\ee
It is called the {\em umbral algebra}.
For $f(t) \in \bC[[t]]$,
$f(D): \bC[x] \to \bC[x]$ is an operator on $\bC[x]$.
It can be characterized as follows.
For $a\in \bR$,
let $E^a: \bC[x] \to \bC[x]$ be
the translation operator defined by:
$$p(x) \mapsto p(x+a).$$
An operator $P$ on $\bC[x]$ is said to be translation
invariant if $PE^a = E^aP$ for all $a\in \bR$.
It is clear that $f(D)$ is translation invariant.
Conversely,
any translation-invariant  operator is of this form.
Therefore,
the umbral algebra is isomorphic to the algebra of translation invariant
linear operators
on $\bC[x]$.

Because both $\bC[x]$ and its dual, $\bC[[t]]$, have structures of
commutative algebras,
one gets a bialgebra structure on $\bC[x]$.
It is called the {\em binomial bialgebra} in one variable \cite{Joni-Rota}.
Note by \eqref{eqn:D-x},
for $f(t) \in \bC[[t]]$, one has
\begin{align} \label{eqn:f-expansion}
f(D) & = \sum_{k=0}^\infty \frac{1}{k!}
\corr{f(D)|x^k} D^k,
\end{align}
it follows that
\ben
f(D)g(D) & = & \sum_{n=0}^\infty \frac{1}{n!}
\sum_{k+l=n}
\binom{n}{k} \corr{f(D)|x^k} \corr{g(D)|x^l} D^n,
\een
and therefore,
\ben
\corr{f(D)g(D)|x^n} & = & \sum_{k+l=n}
\binom{n}{k} \corr{f(D)|x^k} \corr{g(D)|x^l}.
\een
This means the comultiplication on $\bC[x]$ is given by:
\be
\Delta x^n = \sum_{k=0}^n \binom{n}{k} x^k \otimes x^{n-k}.
\ee
Define an isomorphism $\varphi:\bC[x] \otimes \bC[y] \to \bC[x,y]$ by
$$\varphi(x^m \otimes y^n)= x^my^n,$$
then we have:
\be
\varphi(\Delta(x^n)) = (x+y)^n,
\ee
and in general,
\be
\varphi(\Delta(p(x))) = p(x+y).
\ee
The counit is $1 \in \bC[[t]]$:
\be
\corr{1|p(x)} = p(0) = \ev (p(x)).
\ee

In fact, one actually gets a structure of a Hopf algebra on $\bC[x]$.
The antipode is given by:
\be
S(p(x)) = p(-x).
\ee

\subsection{Polynomial sequences of binomial type, coalgebra isomorphisms, and umbral operators}

A coalgebra isomorphism $U$ is a one-to-one onto linear operator on
$\bC[x]$ such that
\be \label{def:UmbralOp}
\Delta U(x^n) = \sum_{k=0}^n \binom{n}{k} U(x^k) \otimes U(x^{n-k}).
\ee
Write $U(x^n) = p_n(x)$.
Such an operator is  called an {\em umbral operator} by Mullin and Rota.
Applying the isomorphism $\varphi$ on both sides of \eqref{def:UmbralOp},
one gets:
\be
p_n(x + y)= \sum_{k=0}^n \binom{n}{k}p_k(x)p_{n-k}(y).
\ee
I.e., the sequence $p_n(x)= U(x^n)$ is a polynomial sequence of binomial type for
 an umbral operator $U$.
Converse,
given a polynomial sequence $\{p_n(x)\}$  of binomial type,
setting $U(x^n)=p_n(x)$ defines an umbral operator (i.e. a coalgebra isomorphism) on $\bC[x]$.

\subsection{Delta series and conjugate sequences}

To study the umbral operators on $\bC[x]$,
it is more natural to study in the dual picture,
i.e., to study isomorphisms of   the dual algebra $\bC[[t]]$.
Given a linear operator $\lambda: \bC[x] \to \bC[x]$,
its adjoint operator $\lambda^*: \bC[[t]] \to \bC[[t]]$ is defined by:
\be
\corr{\lambda^*f(D)|p(x)} = \corr{f(D)|\lambda p(x)}
\ee
for all $p(x) \in \bC[x]$.
The adjoint of an umbral operator is an automorphism of the umbral algebra, and conversely, with
due respect to suitable topology on $\bC[[t]]$.
It follows that $U^*$ is determined by $U^*(t)$ which is of the form
\be
U^*(t) = F(t) = \sum_{n=1}^\infty \frac{a_n}{n!} t^n
\ee
with $a_1 \neq 0$.
Note for $p(x) \in \bC[x]$,
\be \label{eqn:p-expansion}
p(x) = \sum_{k\geq 0} \frac{1}{k!} \corr{D^k|p(x)} x^k,
\ee
Therefore,
\ben
p_n(x) & = & \sum_{k \geq 0}\frac{1}{k!} \corr{D^k|p_n(x)}x^k
=  \sum_{k \geq 0}\frac{1}{k!} \corr{D^k|U(x^n)}x^k \\
& = & \sum_{k \geq 0}\frac{1}{k!} \corr{U^*(D^k)|x^n}x^k
= \sum_{k \geq 0}\frac{1}{k!} \corr{(U^*(D))^k|x^n}x^k \\
& = & \sum_{k \geq 0}\frac{1}{k!} \corr{(U^*(D))^k|x^n}x^k \\
& = & \sum_{k \geq 0}\frac{1}{k!} \corr{F(D)^k|x^n}x^k.
\een
Given a series of the form $F(t) = \sum_{n=1}^\infty \frac{a_n}{n!}t^n$ with $a_1 \neq 0$,
the sequence of polynomials
\be
p_n(x) = \sum_{k \geq 0}\frac{1}{k!} \corr{F(D)^k|x^n}x^k
\ee
is called the {\em conjugate sequence} of the delta series $F(t)$.

\begin{Proposition}
The exponential generating series of the conjugate sequence $\{p_n(x)\}_{n\geq 0}$
of the delta series $F(t)$ is
\be
\sum_{n=0}^\infty p_n(x) \frac{X^n}{n!} = e^{xF(X)}.
\ee
\end{Proposition}

\begin{proof}
We first show that for any series $g(t) = \sum_{k=0}^\infty \frac{b_k}{k!}t^k$,
we have
\be
\corr{g(D)|e^{Xx}} = g(X).
\ee
This is because we  have
\ben
\corr{g(D)|x^n} = \sum_{k\geq 0} \frac{b_k}{k!} \corr{D^k|x^n}
= \sum_{k \geq 0} \frac{b_k}{k!} \cdot k!\delta_{k,n} = b_n,
\een
and so
\ben
\corr{g(D)|e^{Xx}} & = & \sum_{n\geq 0} \corr{g(D)|x^n} \cdot
\frac{X^n}{n!}
= \sum_{n \geq 0} b_n \frac{X^n}{n!} = g(X).
\een
Therefore,
\ben
\sum_{n=0}^\infty p_n(x) \frac{X^n}{n!}
& = &
\sum_{n=0}^\infty \frac{X^n}{n!}  \sum_{k \geq 0}\frac{1}{k!} \corr{F(D)^k|x^n}x^k \\
& = &   \sum_{k \geq 0}\frac{1}{k!} \corr{F(D)^k|\sum_{n=0}^\infty \frac{X^n}{n!} x^n}x^k \\
& = & \sum_{k \geq 0}\frac{1}{k!} \corr{F(D)^k|e^{Xx}}x^k \\
& = & \sum_{k\geq 0} \frac{1}{k!} F(X)^kx^k  = e^{xF(X)}.
\een
\end{proof}

\subsection{Umbral compositions}

Let us return to the point of view that the adjoints of umbral operators are
the elements in the automorphism group of the umbral algebra $\bC[[t]]$.
Given two umbral operators $U$ and $V$, consider the composition $U^*\circ V^*$ of
their adjoints $U^*$ and $V^*$.
One has
\be
\corr{U^*\circ V^* (f(D))|p(x)} = \corr{f(D)|V(U(p(x)))}
\ee
for all $f(t) \in \bC[[t]]$ and $p(x) \in \bC[x]$.
Suppose that
\begin{align}
U(x^n) & = p_n(x) = \sum_{k=0}^n a_{n,k}x^k, &
V(x^n) & = q_n(x), \\
U^*(t) & = F(t) = \sum_{n\geq 1} \frac{a_n}{n!}t^n,
& V^*(t) & = G(t) = \sum_{n\geq 1} \frac{b_n}{n!}t^n,
\end{align}
i.e., $\{p_n(x)\}$ and $\{q_n(x)\}$ are the conjugate sequences of $f(t)$ and $g(t)$
respectively.
Then
\be
r_n(x) = V(U(x^n)) = V(p_n(x)) = \sum_{k=0}^n a_{n,k} V(x^k)
= \sum_{k=0}^n a_{n,k} q_k(x)
\ee
is the conjugate sequence of
\ben
U^*(V^*(t)) & = & U^*(G(t)) = U^*(\sum_{n\geq 1} \frac{a_n}{n!}t^n)
= \sum_{n\geq 1} \frac{a_n}{n!}(U^*(t))^n \\
& = & \sum_{n\geq 1} \frac{a_n}{n!}F(t)^n = G(F(t)).
\een
The sequence $\{r_n(x)\}$ is called the {\em umbral composition} of the sequences
$\{p_n(x)\}$ and $\{q_n(x)\}$.

\subsection{Group structures on $\cS$ in terms of umbral calculus}

Let $\{\gamma_n\}$  and $\{\varphi_n \}$  be two polynomial sequences
of binomial type,
defined by two elements $G(X), F(X) \in \cS$,
respectively.
In other words,
\ben
&& \sum_{n=0}^\infty \gamma_n(x) \frac{X^n}{n!}  = \exp (x G(X)), \\
&& \sum_{n=0}^\infty \varphi_n(x) \frac{X^n}{n!}  = \exp (x F(X)).
\een
Then the sequence $\theta_n$   obtained by
the ``umbral'' substitution of $\varphi_n$   into $\gamma_n$,
i.e.,
\be
\theta_n(x)= \gamma_n (x)|_{x^k \mapsto \varphi_k(x)}
\ee
is a polynomial sequence of polynomial type defined by $F(G(X))$.

\section{Umbral Calculus, Quantum Mechanics, and  Coherent States}

\label{sec:QM}

In this Section we relate umbral calculus to quantum mechanics,
especially to Stone-von Neumann Theorem, Wick quantization, and coherent states.

\subsection{Compositional inverse series and connection coefficients}

Suppose that $\{p_n(x)\}$ and $\{q_n(x)\}$ are the conjugate sequences of $F(t)$ and $G(t)$
respectively.
Then
\be
q_n(x) = \sum_{k=0}^n c_{n,k} p_k(x)
\ee
for some coefficients $c_{n,k}$ called the {\em connection coefficients}.
To find these coefficients,
rewrite the above equations as
\be
V(x^n) = \sum_{k=0}^n c_{n,k} U(x^k),
\ee
where $U$ and $V$ are the umbral operators adjoint to $F(D)$ and $G(D)$ respectively.
Now apply $U^{-1}$ on both sides to get:
\be
(U^{-1}\circ V)(x^n) = \sum_{k=0}^n c_{n,k}x^k.
\ee
This means $\{r_n(x) = \sum_{k=0}^n c_{n,k}x^k\}$ is the conjugate sequence of
\be
(U^{-1}\circ V)^*(t)=V^*((U^{-1})^*(t)) = G(f(t)),
\ee
where $f(t)$ is the {\em compositional inverse series} of $F(t)$, i.e.,
\be
F(f(t)) = f(F(t)) = t.
\ee

\subsection{Delta series and associated sequences}

Now one can apply the inverse umbral map
\be
x^n = U^{-1}(p_n(x))
\ee
in many of the formulas above.
By \eqref{eqn:D-x} we get:
\ben
\corr{D^k|U^{-1}(p_n(x))} = k!\delta_{n,k}.
\een
The left-hand side is equal to
\ben
&& \corr{(U^{-1})^*(D^k)|p_n(x)}
= \corr{f(D)^k|p_n(x)},
\een
so we get:
\be \label{eqn:f-k-p}
\corr{f(D)^k|p_n(x)} = k! \delta_{k,n}.
\ee
Recall $\{p_n(x)\}$ is the sequence of deformations of $\{x^n\}$,
in the dual picture, $\{D^n\}$is deformed to $\{f(D)^n\}$.
In the literature on umbral calculus,
$\{p_n(x)\}$ is called the {\em associated sequence}
of the delta series $f(t)$.

Now we prove the theorems in \cite[Section 2.4]{Roman}
in a similar fashion.

\begin{Proposition} (The Expansion Theorem \cite[Theorem 2.4.1]{Roman} )
For $h(t) \in \bC[[t]]$,
\be \label{eqn:Expansion}
h(D) = \sum_{k=0}^\infty \frac{1}{k!}
\corr{h(D)|p_k(x)} f(D)^k
\ee
where $f(t)$ is a delta series and $\{p_n(x)\}$ is its associated sequence.
\end{Proposition}

\begin{proof}
By \eqref{eqn:f-expansion} we have
\ben
h(D) & = & (U^*)^{-1}(U^*(h(D))
= (U^*)^{-1}\biggl(\sum_{k=0}^\infty \frac{1}{k!}
\corr{U^*(h(D))|x^k} D^k \biggr) \\
& = & \sum_{k=0}^\infty \frac{1}{k!}
\corr{h(D)|U(x^k)} (U^*)^{-1}(D^k) \\
& = & \sum_{k=0}^\infty \frac{1}{k!}
\corr{h(D)|p_k(x)} f(D)^k.
\een
\end{proof}

\begin{Proposition} (The Polynomial Expansion Theorem \cite[Theorem 2.4.2]{Roman})
Let $f(t)$ be a delta series and let $\{p_n(x)\}$ be the associated sequence of $f(t)$.
Then for any $p(x) \in \bC[x]$,
one has
\be
p(x) = \sum_{k\geq 0} \frac{1}{k!}
\corr{f(D)^k|p(x)} p_k(x).
\ee
\end{Proposition}
\begin{proof}
Applying the umbral operator $U$ on both sides of
\eqref{eqn:p-expansion}, we get
\ben
U(p(x)) & = &
\sum_{k\geq 0} \frac{1}{k!} \corr{D^k|p(x)} U(x^k) \\
& = & \sum_{k\geq 0} \frac{1}{k!} \corr{D^k|U^{-1}U(p(x))} p_k(x) \\
& = & \sum_{k\geq 0} \frac{1}{k!}
\corr{(U^{-1})^*(D^k)|U(p(x))} p_k(x) \\
& = & \sum_{k\geq 0} \frac{1}{k!}
\corr{f(D)^k|U(p(x))} p_k(x).
\een
\end{proof}

By induction one gets:
\be \label{eqn:f-k-pn}
f(D)^k p_n(x) = (n)_k \cdot p_{n-k}(x).
\ee

\begin{Proposition} (\cite[Theorem 2.4.5]{Roman})
A sequence $\{p_n(x)\}$ of polynomials is associated to a delta series $f(t)$  if and only if
\bea
&& \corr{1|p_n(x)}= \delta_{n,0},  \label{eqn:(i)} \\
&& f(D)p_n(x) = np_{n-1}(x). \label{eqn:(ii)}
\eea
\end{Proposition}

\begin{proof}
Suppose that $\{p_n(x)\}$ is associated with $f(t)$.
Taking $k=0$ in \eqref{eqn:f-k-p},
one gets \eqref{eqn:(i)}.
To get \eqref{eqn:(ii)},
one uses \eqref{eqn:f-k-p} again to get:
\ben
\corr{f(D)^{k-1}|f(D)p_n(x)} & = & \ev ( f(D)^{k-1}f(D)p_n(x)) \\
& = & \ev(f(D)^kp_n(x))
= \corr{f(D)^k p_n(x)} = k!\delta_{k,n}.
\een
The left-hand side is also equal to
\ben
\corr{(U^{-1})*(D^{k-1})|f(D)p_n(x)} = \corr{D^{k-1}|U^{-1}(f(D)p_n(x))},
\een
so we have
\be
\corr{D^{k-1}|U^{-1}(f(D)p_n(x))} = k! \delta_{k,n}.
\ee
From this one gets:
\be
U^{-1}(f(D)p_n(x)) = n x^{n-1}.
\ee
It follows that
\be
f(D)p_n(x) = nU(x^{n-1}) = n p_{n-1}(x).
\ee
Conversely,
if \eqref{eqn:(i)} and \eqref{eqn:(ii)} hold,
then
\ben
\corr{f(D)^k|p_n(x)} & = & \ev(f(D)^kp_n(x))
= \ev((n)_k p_{n-k}(x)) \\
& = &  (n)_k \cdot \corr{1|p_{n-k}(x)}
= (n)_k \cdot \delta_{n-k,0} \\
& = & n! \cdot \delta_{n,k}.
\een
So one recovers \eqref{eqn:f-k-p}.
\end{proof}

\subsection{Recursion formula for associated sequences and
Heisenberg commutation relation in umbral calculus}

 \begin{Proposition} ( \cite[Theorem 2.4.6]{Roman})
If $\{p_n(x)\}$ is associated to $f(t)$,   then for any $h(t) \in \bC[[t]]$,
\be
h(D)p_n(x) = \sum_{k=0}^n \binom{n}{k} \corr{h(D)|p_k(x)} p_{n-k}(x).
\ee
\end{Proposition}

\begin{proof}
This follows from \eqref{eqn:Expansion} and \eqref{eqn:f-k-pn} as follows:
\ben
h(D)p_n(x) & = & \sum_{k=0}^\infty \frac{1}{k!}
\corr{h(D)|p_k(x)} f(D)^kp_n(x) \\
& = &  \sum_{k=0}^\infty \frac{1}{k!}
\corr{h(D)|p_k(x)} (n)_k p_{n-k} (x) \\
& = & \sum_{k=0}^n \binom{n}{k} \corr{h(D)|p_k(x)} p_{n-k}(x).
\een
\end{proof}

By taking $h(t) = t$, one gets the following corollary:

\begin{Proposition} ( \cite[Theorem 2.4.9]{Roman})  If $\{p_n(x)\}$  is
associated with $f(t)$,
then
\be \label{eqn:pn-prime}
p_n'(x) = \sum_{k=0}^n \binom{n}{k} \corr{D|p_{n-k}(x)} p_k(x).
\ee
\end{Proposition}

\begin{Proposition}
(\cite[Theorem 2.4.8]{Roman})  If $p_n(x)$ is associated to $f(t)$, then
\be \label{eqn:x-pn}
xp_n(x) = \sum_{k=0}^n \binom{n}{k}
\corr{f'(D)|p_{n-k}(x)} p_{k+1}(x).
\ee
\end{Proposition}

\begin{proof}
Write $f(t) = \sum_{k=1}^\infty \frac{b_k}{k!}t^k$.
Then we have:
\ben
f(D) (xp_n(x))
& = &  \sum_{k=1}^\infty \frac{b_k}{k!} D^k (xp_n(x)) \\
& = & x \sum_{k=1}^\infty \frac{b_k}{k!} D^k p_n(x)
+ \sum_{k=1}^\infty \frac{b_k}{k!} \cdot k D^{k-1} p_n(x) \\
& = & x f(D) p_n(x) + f'(D) p_n(x).
\een
By induction we get:
\ben
f(D)^k(xp_n(x)) & = & x f(D)^kp_n(x) + kf'(D)f(D)^{k-1}p_n(x) \\
& = & x f(D)^k p_n(x) + k (n)_{k-1}f'(D) p_{n-k+1}(x).
\een
After evaluating at $x=0$, one gets:
\ben
&& \corr{f(D)^k|xp_n(x)}
= k (n)_{k-1}\corr{f'(D)|p_{n-k+1}(x)}.
\een
By the Polynomial Expansion Theorem,
\ben
xp_n(x) & = & \sum_{k=0}^\infty \frac{1}{k!} \corr{f(D)^k|xp_n(x)} p_k(x) \\
& = & \sum_{k=0}^\infty \frac{1}{k!}\cdot k (n)_{k-1}\corr{f'(D)|p_{n-k+1}(x)} p_k(x) \\
& = & \sum_{k=0}^\infty \binom{n}{k} \corr{f'(D)|p_{n-k}(x)} p_k(x).
\een
\end{proof}

The above two Theorems concern the following two operators on $\bC[x]$: The operator
$D$ defined by $Dp(x) = p'(x)$, and the operator $X$ defined by $Xp(x) = x\cdot p(x)$.
They satisfy the Heisenberg canonical commutation relation:
\be
[D, X] = 1.
\ee
By taking adjoints, we get
\be
[X^*, D^*] = 1.
\ee
These adjoint operators on $\bC[[t]]$ are given as follows:
\begin{align}
 D^*g(t) & = t \cdot g(t), &
 X^*g(t) & = g'(t).
\end{align}

\subsection{Umbral shift and more Heisenberg commutation relations in umbral calculus}

Since we have $p_n(0) = \corr{1|p_n(x)} = 0$ for $n \geq 1$, so one can write
\be
p_n(X) = \sum_{k \geq 1} c_{n,k}x^k, \quad n = 1, 2,\dots.
\ee
Then by \eqref{eqn:pn-prime},
\be
p_n'(x) = \sum_{k=0}^{n-1} c_{n-k,1} p_k(x).
\ee
So after integration one gets a recursion relation:
\be
p_n(x) = \sum_{k=0}^{n-1} c_{n-k,1} \int_0^x p_k(x)dx.
\ee
Therefore,
the sequence $\{p_n(x)\}$ of polynomials of binomial type
is completely determined by the sequence $\{c_{n,1}\}$.
This recovers \cite[p. 697, Corollary 1]{Rota et al}.

By \eqref{eqn:x-pn} we also get a recursion formula:
\be \label{eqn:x-pn-2}
p_{n+1}(x) = \frac{1}{f'(0)} \biggl(xp_n(x) - \sum_{k=0}^{n-1} \binom{n}{k}
\corr{f'(D)|p_{n-k}(x)} p_{k+1}(x) \biggr).
\ee

Recall the generating series of $\{p_n(x)\}$ is
\be \label{eqn:x-Ft}
\sum_{n=0}^\infty \frac{p_n(x)}{n!}t^n
= e^{xF(t)},
\ee
where $F(t) = \sum_{n=1}^\infty \frac{a_n}{n!}t^n$.
Taking $\frac{\pd}{\pd x}$ and then taking $x=0$ on both sides, we get:
\be
\sum_{n=1} \frac{c_{n,1}}{n!} t^n = F(t)
= \sum_{n=1}^\infty \frac{a_n}{n!}t^n.
\ee
It follows that
\be
c_{n,1} = a_n.
\ee
This recovers \cite[p. 697, Corollary 2]{Rota et al}.
In \eqref{eqn:x-Ft}, change $t$ to $D_y = \frac{\pd}{\pd y}$,
let both sides act on $y^n$ and then take $y=0$.
This yields:
\be
p_n(x) = \corr{e^{xF(D_y)}|y^n}.
\ee
From this one can derive a recursion formula as follows:
\ben
p_n(x) & = & \corr{e^{xF(D_y)}|y \cdot y^{n-1}} = \corr{e^{xF(D_y)}| Y (y^{n-1})} \\ \\
& = & \corr{Y^*e^{xF(D_y)}|y^{n-1}}
= \corr{x F'(D_y) e^{xF(D_y)}|y^{n-1}}\\
& = & x \cdot \corr{F'(D_y) e^{xF(D_y)}|U^{-1}(p_{n-1}(y))} \\
& = & x \cdot \corr{(U^{-1})^*(F'(D_y)) (U^{-1})^*(e^{xF(D_y)})|p_{n-1}(y)} \\
& = & x \cdot \corr{F'(f(D_y)) e^{x D_y}|p_{n-1}(y)} \\
& = & x \cdot \corr{F'(f(D_y)) |e^{x D_y}p_{n-1}(y)} \\
& = & x \cdot \corr{\frac{1}{f'(D_y)}|p_{n-1}(x+y)} \\
& = & x\cdot \biggl( \frac{1}{f'(D_y)}p_{n-1}(x+y) \biggr)\biggr|_{y=0} \\
& = & x\cdot \biggl( \frac{1}{f'(D_x)}p_{n-1}(x+y) \biggr)\biggr|_{y=0} \\
& = & x \cdot \frac{1}{f'(D_x)} p_{n-1}(x).
\een
In the above we have used the following facts:
\begin{align*}
(U^{-1})^*D_y & = f(D_y), & F(f(t)) & = t, & F'(f(t)) \cdot f'(t) & = 1.
\end{align*}
So we have recover the following identity in \cite[Corollary 3.6.6]{Roman}:
\be \label{eqn:pn-rec}
p_n(x) = x \cdot \frac{1}{f'(D_x)} p_{n-1}(x).
\ee
The operator $\theta$ on $\bC[x]$ defined by:
\be \label{def:umbral-shift}
\theta_f p_n(x) = p_{n+1}(x)
\ee
is called the {\em umbral shift}.
By \eqref{eqn:pn-rec},
\be
\theta_f = X \circ \frac{1}{f'(D)}.
\ee
By taking adjoint:
\be
\theta_f^* = \frac{1}{f'(t)} \frac{\pd}{\pd t} = \frac{\pd}{\pd f}.
\ee
Note:
\be
f(D)^* g(t) = f(t) \cdot g(t),
\ee
therefore, it is clear that
\be
[\theta_f^*,f(D)^*] = 1,
\ee
and so after taking adjoint,
\be
[f(D), \theta_f] = 1.
\ee
This can be checked directly from \eqref{eqn:(ii)}
and \eqref{def:umbral-shift}.

\subsection{Coherent states in umbral calculus}
\label{sec:Coherent-Umbral}

Let us now interpret
\be
\sum_{n=0}^\infty \frac{p_n(x)}{n!}t^n
= e^{xF(t)}
\ee
from the perspective of quantum mechanics.

First, inspired by \cite[\S 9]{Rota et al},
we introduce a Hermitian product on $\bC[x]$ as follows.
Given a delta series $f(t) = \sum_{n=1}^\infty \frac{a_n}{n!}t^n$,
with associated series $\{p_n(x)\}$,
let $U: \bC[x] \to \bC[x]$ be the umbral operator defined by:
\be
U(x^n) = p_n(x).
\ee
Define
\be
(p(x), q(x)) : = [(U^{-1}p)(f(D_{\bar{x}}))\overline{q(x)}|_{x=0}.
\ee
It is easy to see that
\be
(p_n(x), p_m(x))= \delta_{m,n}n!.
\ee
Furthermore,
\be
(f(D_x)p(x),q(x))=(p(x),\theta_fq(x)).
\ee
Also note
\be
f(D_x)1 = 0,
\ee
so one can take $1$ to be the vacuum vector $\vac$, $a=f(D_x)$ to be the annihilator,
and $a^\dagger = \theta_f$ to be the creator.
Note we have
\begin{align}
a\vac & = 0, [a, a^\dagger] &= 1
\end{align}
and
\be
a e^{xF(t)} =  f(D_x) \sum_{n=0}^\infty \frac{p_n(x)}{n!} t^n
= \sum_{n=0}^\infty \frac{np_{n-1}(x)}{n!}t^n = t \cdot e^{xF(t)}.
\ee
This implies that the family $e^{xF(t)}$ is a family of coherent states indexed by $t$.

\subsection{A brief summary of coherent states}

For the convenience of the reader who does not have a background in quantum mechanics
and coherent states, we summarize some basics  in this Subsection.
We refer to   Berezin-Shubin \cite{Ber-Shu}, Combescure-Robert \cite{Com-Rob}, Hall \cite{Hall}
and Takhtajan \cite{Tak} for proofs and more details.

\subsubsection{Gaussian state as minimum uncertainty state}

In Schr\"odinger representation (position representation),
the state of a quantum system is described by a function $\Psi(x)$ in $x$, such that
$$||\Psi||^2 = \corr{\Psi(x)|\Psi(x)} = \int_{\bR} |\Psi(x)|^2dx = 1.$$
The position and momentum operators are given by the self-adjoint operators:
\begin{align*}
Q & = x\cdot, & P & = - i\hbar \frac{\pd}{\pd x},
\end{align*}
respectively.
They satisfy the canonical commutation relation:
\be
[Q, P] = i\hbar.
\ee
The expectation values of the position and the momentum are defined by
\begin{align*}
x_0 & = \corr{\Psi(x)|Q|\psi(x)},  & p_0 & = \corr{\Psi(x)|P|\psi(x)},
\end{align*}
respectively.
The variances of the position and the momentum are defined by
\begin{align*}
\corr{(\Delta Q)^2} & = \corr{\Psi(x)|(Q-x_0)^2|\psi(x)},  &
\corr{(\Delta P)^2} & = \corr{\Psi(x)|(P-p_0)^2|\psi(x)},
\end{align*}
The uncertainty relation is:
\be
\corr{(\Delta Q)^2} \cdot
\corr{(\Delta P)^2} \geq \frac{\hbar^2}{4}.
\ee
A state is called a {\em minimum uncertainty state} if the equality holds
in the above inequality.

The Gaussian state defined by
\be
\varphi_0(x) = (\pi \hbar)^{-1/4} \exp \biggl( - \frac{x^2}{2\hbar} \biggr)
\ee
is a normalized state with $x_0 = p_0 = 0$.
For this state we have
\ben
&& \corr{\varphi_0(x)|Q^2|\varphi_0(x)}
= \frac{1}{\sqrt{\pi \hbar}} \int_{\bR} x^2 e^{-x^2/\hbar}dx = \frac{\hbar}{2}, \\
&& \corr{\varphi_0(x)|P^2|\varphi_0(x)}
= \frac{1}{\sqrt{\pi \hbar}} \int_{\bR} e^{-x^2/(2\hbar)} (-\hbar^2) \frac{d^2}{dx^2}
e^{-x^/(2\hbar)} dx = \frac{\hbar}{2},
\een
and so we have
\be
\corr{\varphi_0(x)|(\Delta Q)^2|\varphi_0(x)} \cdot
\corr{\varphi_0(x)|(\Delta P)^2|\varphi_0(x)} = \frac{\hbar^2}{4}.
\ee
This means that the Gaussian $\varphi_0(x)$ is a minimum uncertainty state.

\subsubsection{Gaussian state as the ground state  in the system of quantum harmonic oscillator}

The Hamiltonian function of the harmonic oscillator is
\be
H = \frac{p^2}{2} + \frac{x^2}{2}.
\ee
After the canonical quantization
\begin{align}
x & \to \hat{x} = Q= x \cdot, & p & \to \hat{p} = P = -i \hbar \frac{\pd}{\pd x},
\end{align}
one gets the Hamiltonian operator for the quantum harmonic oscillator:
\be
\hat{H} = \half (P^2 + Q^2) = \frac{1}{2}(-\hbar^2 \frac{\pd^2}{\pd x^2} + x^2).
\ee
Define the annihilator $a$ and the creator $a^\dagger$ by
\begin{align}
a & = \frac{1}{\sqrt{2\hbar}}(Q+ iP), & a^\dagger & = \frac{1}{\sqrt{2\hbar}}(Q - iP),
\end{align}
then $a$ and $a^\dagger$ are adjoint to each other,
and
\be
\hat{H} = \hbar \hat{N} + \frac{\hbar}{2},
\ee
where
\be
\hat{N} = a^\dagger a
\ee
is the number operator.
For any normalized $\Psi(x)$,
one has
\ben
&& \corr{\Psi(x)|\hat{H}|\Psi(x)} = \hbar ||a\Psi(x)||^2 + \frac{\hbar}{2} \geq \frac{\hbar}{2},
\een
with equality iff
\be
a \Psi(x) = 0,
\ee
or equivalently,
\be
\hbar \frac{\pd}{\pd x} \Psi(x) = x \cdot \Psi(x).
\ee
It is then clear that the minimum is achieved
when $\Psi(x)$ is equal to $\varphi_0(x)$ up to a phase,
i.e. multiplication by a factor of the form $e^{i\theta}$.

\subsubsection{Coherent state as the generating function of eigenstates  in the system of quantum harmonic oscillator}

Other eigenstates of quantum harmonic oscillator can be obtained by applying the creator
repeatedly:
\be
\varphi_n(x) = (a^\dagger)^n \varphi_0.
\ee
It is easy to see that
\bea
&& a \varphi_n(x) = n \varphi_{n-1}(x), \\
&& (\varphi_m(x),\varphi_n(x)) = m!\delta_{m,n}.
\eea
Physicists prefer to normalize these eigenstates and set:
\be
|n\rangle = \frac{1}{\sqrt{n!}} \varphi_n(x) = \frac{1}{\sqrt{n!}} (a^\dagger)^n\varphi_0(x).
\ee

The exponential generating series of the eigenstates is
\be
\sum_{n=0}^\infty \varphi_n(x) \frac{z^n}{n!}
= e^{za^\dagger} \varphi_0(x).
\ee
Note
\ben
\corr{e^{za^\dagger}\varphi_0(x)|e^{za^\dagger}\varphi_0(x)}
& = & \sum_{m,n=0}^\infty \frac{\bar{z}^mz^n}{m!n!} \corr{\varphi_m(x)|\varphi_n(x)} \\
& = & \sum_{m,n=0}^\infty \frac{\bar{z}^mz^n}{m!n!} \delta_{m,n}
= e^{|z|^2}.
\een
So physicists define the normalized coherent states by:
\be \label{def:z}
|z\rangle:= e^{-|z|^2/2} \sum_{n=0}^\infty \varphi_n(x) \frac{z^n}{n!}
= e^{-|z|^2/2}  e^{za^\dagger} \varphi_0(x).
\ee
It satisfies:
\be
a|z\rangle = z \cdot |z\rangle.
\ee

\subsubsection{Coherent states as minimum uncertainty states}

To understand the properties of $|z\rangle$ one can do the following computations:
\ben
&& \corr{z|Q|z} = \sqrt{\frac{\hbar}{2}}\corr{z|a+a^\dagger|z}
= \sqrt{\frac{\hbar}{2}}(z+\bar{z}), \\
&& \corr{z|P|z} = - i \sqrt{\frac{\hbar}{2}}\corr{z|a-a^\dagger|z}
= -i \sqrt{\frac{\hbar}{2}}(z-\bar{z}).
\een
Similarly,
\ben
&& \corr{z|Q^2|z} = \frac{\hbar}{2}\corr{z|(a+a^\dagger)^2|z}
= \frac{\hbar}{2}\corr{z|a^2+a^\dagger a +a a^\dagger + (a^\dagger)^2|z}\\
& = & \frac{\hbar}{2}\corr{z|a^2+2a^\dagger a +(a^\dagger)^2+1|z}
= \frac{\hbar}{2}[(z+\bar{z})^2+1], \\
&& \corr{z|P^2|z} = - \frac{\hbar}{2}\corr{z|(a-a^\dagger)^2|z}
= -\frac{\hbar}{2}[(z+\bar{z})^2-1].
\een
It follows  that
\ben
&&\corr{(\Delta Q)^2}=\corr{z|Q^2 |z} -\corr{z|Q|z}^2 =  \frac{\hbar}{2},\\
&&\corr{(\Delta P)^2}=\corr{z|P^2 |z} -\corr{z|P|z}^2 = \frac{\hbar}{2},
\een
and so
\ben
&& \corr{z|(\Delta Q)^2|z} \cdot \corr{z|(\Delta P)^2|z} = \frac{\hbar^2}{4}.
\een
This means that $|z\rangle$ is a family of minimum uncertainty states,
with expectation values of the position and the momentum given by
\begin{align} \label{eqn:x0-p0}
x_0 & = \sqrt{\frac{\hbar}{2}}(z+\bar{z}), &
p_0 & = -i\sqrt{\frac{\hbar}{2}}(z-\bar{z}),
\end{align}
respectively.

\subsubsection{Coherent states as Gaussian states acted by Weyl-Heisenberg translations}

Let us now derive an explicit expression for $|z\rangle$.
First, recall that $\varphi_n(x)$ can be expressed in terms of Hermite polynomials:
\be \label{eqn:phi-n}
\varphi_n(x) =(\pi \hbar)^{-1/4} 2^{-n/2}H_n(\frac{x}{\sqrt{h}}) e^{-\frac{x^2}{2\hbar}}
=(\pi \hbar)^{-1/4} \cdot He_n(\frac{\sqrt{2}x}{\sqrt{\hbar}})e^{-\frac{x^2}{2\hbar}},
\ee
where $H_n(x)$ and $He_n(x)$ are the physicists' and probabilists' Hermite polynomials
respectively:
\begin{align*}
H_n(x) &=(-1)^{n}e^{x^{2}}{\frac {d^{n}}{dx^{n}}}e^{-x^2}, &
{\mathit {He}}_{n}(x)=(-1)^ne^{\frac {x^{2}}{2}}{\frac {d^n}{dx^n}}e^{-{\frac {x^2}{2}}}.
\end{align*}
Then one can use the formula for the exponent generating series
of Hermite polynomials:
\begin{align}
\sum_{n=0}^\infty H_n(x) \frac{t^n}{n!} & = e^{2xt-t^2}, &
\sum_{n=0}^\infty He_n(x) \frac{t^n}{n!} & = e^{xt-t^2/2},
\end{align}
together with \eqref{eqn:phi-n} and \eqref{def:z}.
The result is :
\be \label{eqn:z-Explicit1}
|z\rangle
= (\pi \hbar)^{-1/4}e^{-|z|^2/2} \cdot \exp\biggl( -\frac{1}{2\hbar}x^2
+ \frac{\sqrt{2}}{\sqrt{\hbar}} \cdot x z
- \frac{z^2}{2} \biggr).
\ee
By \eqref{eqn:x0-p0},
we have
\be
z = \frac{x_0+ip_0}{\sqrt{2\hbar}}.
\ee
Plugging this into \eqref{eqn:z-Explicit1}, we get:
\be \label{eqn:z-Explicit2}
|z\rangle = (\pi \hbar)^{-1/4} \exp \biggl( - \frac{i}{2\hbar}x_0p_0\biggr)
\exp\biggl(\frac{i}{\hbar} p_0x \biggr)\exp\biggl(-\frac{(x-x_0)^2}{2\hbar} \biggr).
\ee

If one introduces the Weyl-Heisenberg translation operator:
\be \label{eqn:T1}
\hat{T}(z) = e^{-\frac{i}{2\hbar} x_0p_0} e^{ \frac{i}{\hbar}p_0 \cdot Q}
e^{- \frac{i}{\hbar} x_0 \cdot P},
\ee
then one has
\be \label{eqn:z-T-phi}
|z\rangle = \hat{T}(z)\varphi_0.
\ee
Using the Heisenberg commutation relation
\be
[Q, P] = i \hbar,
\ee
one can show that
\be \label{eqn:T2}
\hat{T}(z) = \exp \biggl( \frac{i}{\hbar} (p_0Q - x_0 P)\biggr).
\ee

Another way to obtain \eqref{eqn:z-T-phi} and hence also \eqref{eqn:z-Explicit2}
is to rewrite
Weyl-Heisenberg translations in terms of creator and annihilator:
\be
\hat{T}(z) = \exp \big(\alpha \cdot a^\dagger - \bar{\alpha}\cdot a \big)
= \exp \biggl(-\frac{|\alpha|^2}{2}\biggr)
\exp\big(\alpha \cdot a^\dagger\big) \cdot \exp\big(- \bar{\alpha} \cdot a\big).
\ee
Then because
\be
a\varphi_0 = 0,
\ee
one has
\be
 \hat{T}(z)\varphi_0
 = \exp \biggl(-  \frac{|\alpha|^2}{2}\biggr)
\exp\biggl( \alpha \cdot  a^\dagger\biggr)\varphi_0 = |z \rangle.
\ee

\subsubsection{Coherent states and representation of the Weyl-Heisenberg group}

The following two formulas for the multiplication of the operators $\hat{T}(z)$ with $\hat{T}(z')$
is well-known:
\bea
\hat{T}(z) \hat{T} (z') = \exp \biggl(- \frac{i}{2\hbar}
\sigma(z,z')  \biggr) \hat{T} (z + z'), \\
\hat{T} (z) \hat{T} (z') = \exp \biggl(- \frac{i}{\hbar}
\sigma(z,z') \biggr) \hat{T}(z')\hat{T} (z),
\eea
where for $z = q+ip$, $z' = q'+ip'$, $ \sigma(z,z')$ is the symplectic product:
\be
\sigma(z,z') = q \cdot p' - p \cdot q'.
\ee
These multiplication formulas lead to the definition of the
 Heisenberg group $H_1$.
 It is the space of $(t, z) \in \bR \times \bC$ with the group multiplication
\be
(t,z)(t',z') = (t + t' + \sigma(z,z'),z + z').
\ee
It is a Lie group of dimension $3$.
The Schr\"odinger representation is defined as the following
irreducible unitary representation
of $H_1$ in $L^2(\bR)$:
\be
\rho(t,z) = e^{-it/2} e^{\hat{T} (z)}.
 \ee

\subsubsection{Stone-von Neumann Theorem}

Suppose that $\tilde{Q}$ and $\tilde{P}$ are two self-adjoint operator on a Hilbert space $H$ ,
such that the following commutation relation is satisfied :
\be
[\tilde{Q}, \tilde{P}] = i,
\ee
then one can define a representation of the Heisenberg group $H_1$ as follows:
\be
\tilde{\rho}(t,z) = e^{-it/2} \exp \biggl( \frac{i}{\hbar} (p_0\tilde{Q} - x_0 \tilde{P})\biggr),
\ee
where $x_0$ and $p_0$ are related to $z$ by \eqref{eqn:x0-p0}.
The Stone-von Neumann Theorem states that when the representation $\tilde{\rho}$ is irreducible,
there is a unique unitary map $U: H \to L^2(\bR)$ such that
\be
U \tilde{\rho}(t,z) U^{-1} = \rho(t,z).
\ee

\subsubsection{Weyl quantization}

The Lie algebra of the Heisenberg group $H_1$ is the Heisenberg algebra $\fh_1$,
generated by $e_1, e_2, e_3$, with the following commutation relations:
\begin{align}
[e_1, e_2]& = e_3, &
[e_3, e_1] &= [e_3, e_2] = 0.
\end{align}
By the PBW Theorem,  elements of the form
$e_3^{n_3}e_2^{n_2} e_1^{n_1}$
where each $n_k$ is a non-negative integer, span
the universal enveloping algebra $U\fh_1$ and are linearly independent.
The same statement also holds for elements of the form
$e_3^{n_3}e_1^{n_1} e_2^{n_2}$.
Another choice for the basis of $U\fh_1$ is as follows:
For two expressions $X$ and $Y$, consider the expansion:
\be
(aX+bY)^n = \sum_{k=0}^n \binom{n}{k} C^n_{k,n-k}(X,Y)a^kb^{n-k},
\ee
then elements of the form $$\frac{n_1!n_2!}{(n_1+n_2)!}e_3^{n_3}C^{n_1+n_2}_{n_1,n_2}(e_1,e_2)$$
form a basis of $U\fh_1$.

The three different bases of $U\fh_1$ corresponds to
three  different quantizations: the
$qp$-symbol, the $pq$-symbol and the Weyl symbol.
Given an operator of the form:
\be
f(q,p) = \sum_{m,n\geq 0} c_{m,n} q^mp^n,
\ee
its $qp$-quantization, $pq$-quantization, and Weyl quantization are given by
\bea
&& \hat{f}_{qp}= \sum_{m,n\geq 0} c_{m,n} Q^mP^n, \\
&& \hat{f}_{pq}= \sum_{m,n\geq 0} c_{m,n} P^nQ^m, \\
&& \hat{f}_{Weyl} = \sum_{m,n\geq 0} c_{m,n}
\frac{m!n!}{(m+n)!}C^{m+n}_{m,n}(Q,P)
\eea
respectively.
Conversely,
$f(q,p)$ is called
the $qp$-symbol of $\hat{f}_{qp}$,
the $pq$-symbol of $\hat{f}_{pq}$,
and the Weyl symbol of $\hat{f}_{Weyl}$.

The transformation from $p^n$ to $P^n = (-i\hbar\frac{\pd}{\pd x})^n$
 can be achieved
by Fourier transform and inverse Fourier transform:
\bea
&& \tilde{f}(p)=\frac{1}{\sqrt{2\pi}} \int_{-\infty }^{\infty } f(x)\ e^{-ix p }\,dx, \\
&& f(x)=\frac{1}{\sqrt{2\pi}} \int_{-\infty }^{\infty }\tilde{f}(p)\ e^{ix p }\,dp.
\eea
First note:
\ben
(-i\hbar \frac{\pd}{\pd x})^n \psi(x)
& = & \frac{1}{\sqrt{2\pi}}  \int_{-\infty }^{\infty } (\hbar p)^n
\tilde{\psi}(p)\ e^{ ix p }\,dp \\
& = & \frac{1}{2\pi\hbar} \int_{-\infty}^\infty\int_{-\infty}^\infty p^n
e^{i(x-y)p/\hbar} \psi(y)dydp,
\een
so we have:
\be \label{eqn:hat-f}
\hat{f}_{qp}\psi(x) = \frac{1}{2\pi\hbar} \int_{-\infty}^\infty\int_{-\infty}^\infty f(x,p)
e^{i(x-y)p/\hbar} \psi(y)dydp.
\ee
If we write
\be
\hat{f}_{qp}\psi(x) =  \int_{-\infty}^\infty K(x,y) \psi(y)dy,
\ee
then the kernel function $K(x,y)$ can be obtained from the symbol $f(q,p)$ by Fourier transform:
\be \label{eqn:K-in-f}
K(x,y) = \frac{1}{2\pi \hbar} \int_{-\infty}^\infty
f(x,p)e^{-ip (y-x)/\hbar} dp.
\ee
Conversely, the symbol $f(q,p)$ can be obtained from the kernel function
by inverse Fourier transform:
\be \label{eqn:f-in-K}
f(q,p) = \int_{-\infty}^\infty K(q,y) e^{ip(y-q)/\hbar}dy.
\ee
Given two $qp$-symbols $f_1$ and $f_2$,
consider the composition operator $\hat{f}_1\hat{f}_2$.
Denote by $K$, $K_1$, $K_2$ the kernels of the operators $\hat{f}_{1,qp}\hat{f}_{2,qp}$,
$\hat{f}_1$, $\hat{f}_2$. Then
\be
K(x,y) = \int_{-\infty}^\infty K_1(x,z)K_2(z,y)dz.
\ee
Using \eqref{eqn:K-in-f} and \eqref{eqn:f-in-K},
one can see that $\hat{f}_{1,qp}\hat{f}_{2,qp}$ is $\hat{f}_{qp}$,
where
\ben
f(q,p) = \frac{1}{(2\pi \hbar)^2}
\int_{\bR^4} && \exp \frac{i}{\hbar}(  p(y - q) - p_1(z - q) - p_2 (y-z))\\
&& \cdot f_1(q,p_1)f_2(z, p_2)dp_1 dp_2 dydz.
\een
Using the formula
\be \label{eqn:Fourier-Delta}
\frac{1}{2\pi \hbar} \int_{-\infty}^\infty e^{i(x-q)y/h} dy = \delta(x - q),
\ee
one can perform the integration in $y$ to get:
\ben
&& f(q,p) = \frac{1}{2\pi \hbar}
\int_{\bR^2} \exp\biggl[ -\frac{i}{\hbar}
(p_1-p)(q_1 - q)\biggr]
\cdot f_1(q,p_1)f_2(q_1, p)dp_1 dq_1.
\een
Now note:
\ben
&&  \frac{1}{2\pi \hbar}
\int_{\bR^2} \exp\biggl[ -\frac{i}{\hbar}
(p_1-p)(q_1 - q)\biggr]
\cdot  p_1^kq_1^ldp_1 dq_1 \\
& = & \frac{1}{2\pi \hbar}
\int_{\bR} q_1^l e^{\frac{i}{\hbar}p(q_1-q)} \biggl(
\int_\bR e^{ -\frac{i}{\hbar}p_1(q_1 - q)}
\cdot  p_1^k dp_1 \biggr) dq_1 \\
& = & \frac{1}{2\pi}
\int_{\bR} q_1^l e^{\frac{i}{\hbar}p(q_1-q)} \cdot
2\pi \biggl(\frac{i}{\hbar}\bigg)^k \delta^{(k)}(q_1-q)dq_1 \\
& = &  \frac{1}{2\pi}
\int_{\bR} (-1)^k \frac{\pd^k}{\pd q_1^k} \biggl(q_1^l e^{\frac{i}{\hbar}p(q_1-q)} \biggr)\cdot
2\pi (i\hbar)^k \delta(q_1-q)dq_1 \\
& = & \frac{1}{2\pi}
\int_{\bR} (-1)^k \sum_{j=0}^k (l)_j q_1^{l-j} (\frac{i}{\hbar}p)^{k-j} e^{\frac{i}{\hbar}p(q_1-q)}  \cdot
2\pi (i\hbar)^k \delta(q_1-q)dq_1 \\
& = & (-i\hbar)^k \sum_{j=0}^k (l)_j q^{l-j} (\frac{i}{\hbar}p)^{k-j} \\
& = & (p - i\hbar\frac{\pd}{\pd q_1})q_1^l|_{q_1=q}.
\een
So one gets:
\be
f(q,p) = f_1(q,p-i\hbar\frac{\pd}{\pd q_1})f_2(q_1,p)\biggl|_{q_1=q}.
\ee
From this one gets:
\bea
f(q,p)& = &\sum_{n=0}^\infty \frac{1}{n!}
\biggl(\frac{\hbar}{i} \biggr)^n
\pd_p^nf_1(q,p) \cdot \pd_q^n f_2(q,p) \\
& = & \exp \biggl(-i\hbar
\frac{\pd^2}{\pd p_1\pd q_1} \biggr)
[f_1(q,p_1)f_2(q_1,p)]
 \biggl|_{q_1=q,p_1=p}.
\eea
Note
\ben
&& \frac{1}{(2\pi)^2} \int_{\bR^2} e^{ir Q} e^{isP}\psi(x)
\int q^m e^{-irq}dq \cdot \int p^n e^{-isp} dp \cdot drds \\
 & = &
\frac{1}{(2\pi)^2} \int_{\bR^2} e^{ir Q} e^{isP}\psi(x)
2\pi i^m\delta^{(m)}(r)\cdot 2\pi i^n\delta^{(n)}(s)drds \\
& = & \frac{1}{(2\pi)^2} \int_{\bR^2} e^{ir x} \psi(x+s\hbar) 2\pi i^m \delta^{(m)}(r)
\cdot 2\pi i^n \delta^{(n)}(s)drds \\
& = & (-i)^{m+n}
\int_{\bR^2} \frac{\pd^m}{\pd r^m} \frac{\pd^n}{\pd s^n}
(e^{irx}\psi (x+s\hbar)) \delta(r)\delta(s)  drds \\
& = & x^m \cdot (-\hbar \frac{\pd}{\pd x})^n \psi(x) \\
& = & (\widehat{(q^mp^n)}_{qp}\psi)(x).
\een
So one gets:
\be
(\hat{f}_{qp}\psi)(x) = \frac{1}{(2\pi)^2}
\int_{\bR^4} f(p,q) e^{ir (Q-q)} e^{is(P-p)}\psi(x) dpdq drds.
\ee

The Weyl quantization is given by
\be
(\hat{f}_{Weyl}\psi)(x) = \frac{1}{(2\pi)^2}
\int_{\bR^4} f(p,q) e^{ir (Q-q)+is(P-p)}\psi(x) dpdq drds.
\ee
By \eqref{eqn:T1} and \eqref{eqn:T2},
\be
[e^{i(rQ+sP)}\psi] (x)
= e^{i(rx+\hbar rs/2)}\psi(x +\hbar s),
\ee
whence
\be
\begin{split}
(\hat{f}_{Weyl}\psi)(x) = & \frac{1}{(2\pi)^2}
\int_{\bR^4} \exp \biggl[i(r x + \hbar r s/2)\biggr] \\
& \cdot e^{-i(rq+sp)}f(q,p)\psi(x + \hbar s)dpdq dr ds.
\end{split}
\ee
By changing $s$ to $\frac{1}{\hbar}(y-x)$,
integrating over $r$, then integrating over $q$,
one gets:
\be
(\hat{f}_{Weyl}\psi)(x)
= \frac{1}{2\pi \hbar} \int_{\bR^2}
e^{\frac{i}{\hbar}(x-y)p} f(\frac{x+y}{2},p)
\psi(y)dy dp.
\ee
The kernel of this operator is
\be
K(x,y) = \frac{1}{2\pi \hbar} \int_\bR
e^{\frac{i}{\hbar}(x-y)p} f(\frac{x+y}{2},p)dp.
\ee
Note K(x,y) is obtained as the Fourier transform
$p\mapsto x-y$ of the function
$f(\frac{x+y}{2},p)$.
Using the inversion formula,
one gets:
\be
f(q,p) =
\int_\bR K(q - \xi/2, q + \xi/2)e^{ip\xi /\hbar} d\xi.
\ee
Using the above two formulas,
one can derive the composition formula for Weyl symbols.
The result is:
\be
f(q,p) =
(f_1 *f_2)(q,p) = e^{i\hbar L} [f_1(q_1,p_1)f_2(q_2,p_2)]
\biggl|_{q_1=q_2=q,p_1=p_2=p},
\ee
where
\be
L=\frac{1}{2}\biggl(
\frac{\pd^2}{\pd q_1\pd p_2}
- \frac{\pd^2}{\pd q_2\pd p_1} \biggr).
\ee
This leads to deformation quantization.

By comparing this summary with the formalism of the umbral calculus,
one can see that the umbral calculus can be embedded in the formulism of quantum mechanics.
More precisely,
one can regard the formal power series $\sum_{n=0}^\infty \frac{a_n}{n!}t^n$ as the symbol
of the operator $\sum_{n=0}^\infty \frac{a_n}{n!}D^n$, and so the latter is the quantization
of the former by any of the above three quantization schemes.
In  the next subsection,
we will see that it is more natural to interpret this quantization as the Wick quantization.

\subsection{Segal-Bargmann transform, Wick quantization and umbral calculus}

In Section \ref{sec:Coherent-Umbral} we have interpreted the exponential generating series of polynomial sequences
of binomial type as coherent states of some quantization.
To elaborate this interpretation further,
in this Subsection we will relate umbral calculus to further results about coherent states
in quantum mechanics.

\subsubsection{Overecompleteness of coherent states}

Given $\psi(x) \in L^2(\bR)$,
one has a decomposition:
\be
\psi(x) = \sum_{n=0}^\infty  a_n \varphi_n(x),
\ee
where the coefficients $a_n$ are given by:
\be
a_n = \frac{1}{n!} (\psi(x), \varphi_n(x)).
\ee
Because
\be
(\varphi_n(x), \varphi_z(x)) = \corr{z|\varphi_n(x)} = e^{-|z|^2/2} \bar{z}^n ,
\ee
one gets:
\be
(\psi(x), \varphi_z(x)) =  \corr{z|\varphi(x)}
= e^{-|z|^2/2}\sum_{n=0}^\infty a_n  \bar{z}^n.
\ee
Furthermore,
\ben
\int_\bC \corr{z|\varphi_m(x)} |z\rangle \frac{dz\wedge d\bar{z}}{-2i}
& = & \int_\bC e^{-|z|^2/2}\frac{\bar{z}^m}{m!}\cdot
e^{-|z|^2/2}\sum_{n=0}^n \varphi_n(x)\frac{z^n}{n!} \frac{dz \wedge d\bar{z}}{-2i} \\
& = & \varphi_m(x),
\een
and so
\be \label{eqn:Over-Complete}
\int_\bC \corr{z|\psi(x)} |z\rangle \frac{dz\wedge d\bar{z}}{-2i}
= \psi(x).
\ee
Indeed,
\ben
\int_\bC \corr{z|\psi(x)} |z\rangle \frac{dz\wedge d\bar{z}}{-2i}
& = & \int_\bC \corr{z|\sum_{n=0}^\infty a_n \varphi_n(x)} |z\rangle \frac{dz\wedge d\bar{z}}{-2i} \\
& = & \sum_{n=0}^\infty a_n \int_\bC \corr{z| \varphi_n(x)} |z\rangle \frac{dz\wedge d\bar{z}}{-2i} \\
& = & \sum_{n=0}^\infty a_n \varphi_n(x) = \psi(x).
\een

\subsubsection{Segal-Bargmann space and Segal-Bargmann transform}

The Fourier-Bargmann transform is defined by the following formula:
\be
F^B_u v(z) =: v^\sharp (z) = (2\pi\hbar)^{-1/2} \corr{u_z,v}.
\ee
It is an isometry from $L^2(\bR)$ into $L^2(\bR^2)$:
 the scalar product of two states $\psi$, $\psi' \in  L^2(\bR)$
can be expressed in terms of $\psi^\sharp, (\psi')^\sharp$:
\be
\corr{\psi',\psi}
 = \int_\bC  (\psi')^\sharp(z)\psi^\sharp(z) \cdot \frac{dz \wedge d\bar{z}}{-2i}.
\ee
Using the Fourier-Bargmann transform,
the overcompleteness of coherent states \eqref{eqn:Over-Complete} can be written as:
\be
\psi(x) = \int_{\bC} \psi^\sharp(z)\varphi_z(x) \frac{dz\wedge d\bar{z}}{-2i}.
\ee

Let
\be
\xi = \frac{x_0-ip_0}{\sqrt{2}}.
\ee
The Segal-Bargmann space is the space of holomorphic functions $F$ in $\xi$
such that
\be
||F||_\hbar:= \biggl(\frac{1}{\pi} \int_\bC |F(\xi)|^2 e^{-|\xi|^2/\hbar} \frac{d\xi\wedge d\xi}{-2i}\biggr)^{1/2}
< +\infty.
\ee
This is a Hilbert space in which the space of polynomials in $\xi$ is dense.

The transformation $\psi \mapsto \psi^\sharp_{Hol}(\xi):=\psi^\sharp(z) e^{|z|^2/2}$ is called the
Segal-Bargmann transform
and is denoted by $\cB$.
One can check that:
\bea
&& \cB(x\psi)(\xi) = \frac{1}{\sqrt{2}}(\hbar\pd_\xi + \xi)\cB\psi(\xi), \\
&& \cB(\hbar \pd_x\psi)(\xi) = \frac{1}{\sqrt{2}}(\hbar\pd\xi - \xi)B\psi(\xi), \\
&& \cB [a^\dagger\psi] (\xi) = \xi\cB[\psi](\xi) , \\
&& \cB[a\psi](\xi) = \pd_\xi\cB[\psi](\xi).
\eea
The integral kernel of the Bargmann transform $\cB$ is the Bargmann kernel:
\be
\cB(x,\xi) = (\pi \hbar)^{-3/4}2^{-1/2} \exp\biggl[-\frac{1}{\hbar}
\biggl( \frac{x^2}{2} - \sqrt{2}x \cdot \xi + \frac{\xi^2}{2} \biggr)
\biggr],
\ee
where $x \in  \bR$, $\xi \in  \bC$.
The Bargmann kernel is a generating function for the Hermite
functions:
\be
B(x,\xi) =
\sum
\frac{\xi^n}{((2\pi)^n n!)^{1/2}} \phi_n(x).
\ee
The integral kernel of $\cB^{-1}$ is
\be
\cB^{-1}(\xi,x) = (\pi \hbar)^{-31/4}2^{-n/2} \exp
\biggl[-\frac{1}{\hbar} \biggl( \frac{x^2}{2} - \sqrt{2}x \cdot \bar{\xi} + \frac{\bar{\xi}^2}{2}
\biggr) \biggr].
\ee

\subsubsection{Wick symbol, Wick quantization and Wick $*$-product}

For an operator $A$ in the Wick normal form
\be
\hat{A}_{Wick}= \sum_{k,l} A_{k,l}(a^\dagger)^k a^l,
\ee
define its Wick symbol by
\be
A(\xi, \bar{\xi}) = \sum_{k,l} A_{k,l}z^k\bar{z}^l.
\ee
The operator $\hat{A}_{Wick}$is called the Wick quantization of the Wick symbol $A(\xi, \bar{\xi})$.
The Wick symbol can be obtained by the following formula:
\be
A(\xi,\bar{\xi}) = \frac{(\hat{A}_{Wick}\varphi_\xi, \varphi_\xi)}{(\varphi_\xi, \varphi_\xi)}.
\ee
Using the Segal-Bargmann transform,
one can identify $\hat{W}_{Wick}$ with an operator on the Segal-Bargmann space,
also denoted by $\hat{A}_{Wick}$:
\be
(\hat{A}_{Wick}f)(\xi) = \frac{1}{\pi} \int_\bC
A(\xi, \bar{\eta}) f(\eta) e^{-\bar{\eta}(\eta-\xi)}\frac{d\eta\wedge d\bar{\eta}}{-2i}.
\ee
If $A_1(\xi,\bar{\xi})$ and $A_2(\xi,\bar{\xi})$ are the Wick symbols of operators $\hat{A}_{1,Wick}$
and $\hat{A}_{2,Wick}$, then the Wick symbol of the operator
$\hat{A}_{Wick} =\hat{A}_{1,Wick}\hat{A}_{2,Wick}$ is given by
\be
A(\xi,\bar{\xi}) = \frac{1}{\pi}
\int_\bC A_1(\xi,\bar{\eta})A_2(\eta,\bar{\xi})e^{-(\eta-\xi)(\bar{\eta}-\bar{\xi})}
\frac{d\eta\wedge d\bar{\eta}}{-2i}.
\ee
Write
$A(\xi,\bar{\xi}) = A_1(\xi,\bar{\xi})*_{Wick}A_2(\xi,\bar{\xi})$.
The following formula holds:
\be
A_1(\xi,\bar{\xi})*_{Wick}A_2(\xi,\bar{\xi})
= \exp \biggl( \hbar \frac{\pd^2}{\pd \eta \pd \bar{\eta}} \biggr) (A_1(\xi, \bar{\eta})A_2(\eta, \bar{\xi}))
\biggl|_{\eta=\xi, \bar{\eta} = \bar{\xi}}.
\ee

\subsubsection{Umbral calculus from the point of view of quantum mechanics}

Now we can interpret umbral calculus using Stone-von Neumann Theorem.
In Section \ref{sec:Coherent-Umbral},
given a delta series in the sense of \cite{Roman},
we have constructed an inner product on the space $\bC[\xi]$ of polynomials in $\xi$,
on which the operator $f(D_\xi)$ is interpreted as the annihilator $a$ and its adjoint $\theta_f$
is interpreted as the creator.
The umbral operator $U: \bC[\xi] \to \bC[\xi]$, $\xi^n \mapsto p_n(\xi)$,
can now be interpreted as the intertwining operator from the Segal-Bargmann space
to the above space, as predicted by the Stone-von Neumann Theorem.
As already mentioned in Section \ref{sec:Coherent-Umbral},
the exponential generating series of the associated sequence of polynomials
can be identified with the image of the coherent states in the Segal-Bargmann space
under the map $U$.

One can use the above interpretation to define an umbral $*$-product as follows:
\be
g(x)*h(x) := U(U^{-1}(g(x))\cdot U^{-1}(h(x))).
\ee
More explicitly,
\be
\sum_{k=0}^\infty a_k p_k(x) * \sum_{l=0}^\infty b_l p_l(x)
:= \sum_{k,l=0}^\infty a_k b_lp_{k+l}(x).
\ee
For a special case of this product, see \cite[Appendix]{Dimakis}.

\subsection{Segal-Bargmann transform from the point of view of umbral calculus}

The Segal-Bargmann transform is essentially a linear map
\be
\bC[x] \to \bC[x], He_n(x) \mapsto x^n.
\ee
The inverse map is then given by
\be
x^n \mapsto He_n(x).
\ee
From the point of view of umbral calculus,
this is a Sheffer operator.
We will recall in this Subsection the theory of Sheffer sequences
in umbral calculus.
Since it is very similar to the theory of associated sequences,
we will be very brief and refer the interested reader to
Roman \cite{Roman} for details.

\subsubsection{Sheffer sequences}
The reference for this subsection is \cite[Section 2.3]{Roman}.
Let $f(t)$ be a delta series and let $g(t)$ be an invertible series.
Then there exists a unique sequence $\{s_n(x)\}_{n\geq 0}$  of polynomials satisfying
\be
\corr{g(t)f(t)^k|s_n(x)}  = n! \delta_{n,k}
\ee
for all $n, k  \geq 0$.
This sequence is called the {\em Sheffer sequence}. for the pair
$(g(t),f(t))$.
Similar to the associated sequence,
one has the Expansion Theorem and the Polynomial Expansion Theorem for Sheffer sequences.
There are several ways to characterize the Sheffer sequences:
The following statements are equivalent to each other:
\begin{itemize}
\item[(a)] $\{s_n(x)\}_{n \geq 0}$ is a Sheffer sequence
 for the pair $(g(t), g(t))$.
 \item[(b)]  $\{s_n(x)\}_{n \geq 0}$  satisfies the
Sheffer identity:
\be
s_n(x+y) = \sum_{k=0}^n \binom{n}{k} p_k(y) s_{n-k}(x),
\ee
where $\{p_n(x)\}_{n \geq 0}$ is associated to $f(t)$.
\item[(c)] The exponential generating series of $\{s_n(x)\}_{n\geq 0}$ is given by:
\be \label{eqn:Sheffer-Gen}
\sum_{n=0}^\infty \frac{s_n(x)}{n!}t^n
=\frac{1}{g(F(t))}e^{xF(t)},
\ee
where $F(t)$ is the compositional inverse of $f(t)$.
\item[(d)] For $n \geq 0$,
\be
s_n(x) = \sum_{k=0}^n \corr{f(F(t))^{-1}f(t)^k|x^n} x^k.
\ee
\item[(e)] $\{g(D)s_n(x)\}_{n\geq0}$ is the associated sequence for $f(t)$.
\item[(f)] For $n \geq 0$,
\be
f(D) s_n(x) = n s_{n-1}(x).
\ee
\end{itemize}

By comparing (f) with \eqref{eqn:(i)} and \eqref{eqn:(ii)},
one is led to the problem finding $g(t)$ from the Sheffer sequence.
This can be solved as follows.
By \eqref{eqn:Sheffer-Gen} one has:
\ben
&& \sum_{n=0}^\infty \frac{s_n'(x)}{n!}t^n
=\frac{F(t)}{g(F(t))}e^{xF(t)},
\een
and so
\be
\sum_{n=0}^\infty \frac{s_n'(0)}{n!}f(t)^n
=\frac{t}{g(t)}.
\ee

\subsubsection{Sheffer operators and umbral compositions}

The materials in this subsection is taken from
\cite[Section 3.5]{Roman}.
Let $\{s_n(x)\}$  be the Sheffer sequence for $(g(t),f(t))$. Then the linear operator
$\lambda_{f,g}$ on $\bC[x]$ defined by
\be
\lambda_{g,f} x^n  = s_n(x)
\ee
is called the {\em Sheffer operator}  for $\{s_n(x)\}$  or for $( g ( t ) , f ( t ) )$.
The adjoint operator of a Sheffer operator can be characterized as follows.
Let $\{p_n(x)\}$ be the sequence of  polynomials associated  to $f(t)$,
and denote by $\lambda_f: \bC[x] \to \bC[x]$ the umbral operator defined by $\lambda_fx^n = p_n(x)$.
Then the Sheffer operator $\lambda_{g,f}$ is related to $\lambda_f$ by:
\be
\lambda_{g,f}= \frac{1}{g(D)}\lambda_f.
\ee
For $h(t) \in \bC[[t]]$,
\be
\lambda_{g,f}^*h(t)= \lambda_f^*\frac{1}{g(D)^*}h(t) = \lambda^*_f(\frac{1}{g(t)}h(t))
= \frac{1}{g(F(t))}h(F(t)).
\ee
Conversely, an operator on $\bC[[t]]$ of the form $h(t) \mapsto \frac{1}{g(F(t))}h(F(t))$
is a Sheffer operator,
where $F(t)$ is a delta series inverse to $f(t)$ and $g(t)$ is an invertible series.

From the above characterization of Sheffer operator one can deduce that
the set of all Sheffer operators is a group under composition:
\bea
&& \lambda_{g(t),f(t)}\circ \lambda_{h(t),l(t)}  = \lambda_{g(t)h(f(t)), l(f(t))}, \\
&& \lambda_{g(t),f(t)}^{-1} = \lambda_{1/g(F(t)), F(t)}.
\eea
As corollaries one has the following results.
A  Sheffer operator maps Sheffer sequences to Sheffer
sequences.
If $\{s_(x)\}$  and $\{r_n(x)\}$ are Sheffer sequences,
then the linear operator defined by: $s_n(x)  \mapsto r_n(x)$ is a Sheffer operator.

\subsubsection{Sheffer shifts and canonical commutation relation}

Let $\{s_n(x)\}$  be the Sheffer sequence for $(g(t),  f(t))$.
The linear operator $\theta_{g,f}$ on $\bC[x]$ defined by
\be
\theta_{g,f}s_n(x)  = s_{n+l}(x)
\ee
is called the {\em Sheffer shift} for $\{s_n(x)\}$ or for $(g(t), f(t))$.

If $\{s_n(x)\}$  is Sheffer for $(g(t),  f (t ))$,  then $\{p_n(x)  = g(D)s_n(x)\}$ is associated to $f(t)$.
and so for $n \geq 0$,
\ben
&& \theta_{g,f}g(D)^{-1}p_n(x) = g(D)^{-1}p_{n+1}(x) = g(D)^{-1} \theta_fp_n(x).
\een
Thus,
\be
\theta_{g,f} = g(D)^{-1}\theta_fg(D).
\ee
By a computation using the chain rule,
one gets:
\be
\theta_{g,f} = \biggl[x \cdot - \frac{g'(D)}{g(D)} \biggr] \frac{1}{f'(D)},
\ee
and in particular,
one gets the following recursion formula for the Sheffer sequence:
\be
s_{n+1}(x) = \biggl[x \cdot - \frac{g'(D)}{g(D)} \biggr] \frac{1}{f'(D)}s_n(x).
\ee

It is clear that the following commutation relation is satisfied:
\be
[f(D_x), \theta_{g,f}] = 1.
\ee

\subsubsection{Exponential generating series of Sheffer sequence as coherent states}

As in \S \ref{sec:Coherent-Umbral},
we modify the scalar product introduced in \cite[\S 9]{Rota et al}
to a Hermitian inner product.
For $p(x), q(x) \in \bC[x]$,
define
\be
(p(x),q(x)):= ((\theta_{g,f}^{-1}p)(f(D))g(D)\overline{q(x)}|_{x=0}.
\ee
It is easy to see that
\be
(s_n(x), s_m(x))_{f,g}= \delta_{m,n}n!.
\ee
Furthermore,
\be
(f(D_x)p(x),q(x))_{f,g}=(p(x),\theta_{g,f}q(x))_{f,g}.
\ee
Also note
\be
f(D_x)1 = 0,
\ee
so one can take $1$ to be the vacuum vector, $a=f(D_x)$ to be the annihilator,
and $a^\dagger = \theta_{g,f}$ to be the creator.
Note we have
\ben
a\biggl[ g(t)e^{xF(t)} \biggr] & = & f(D_x) \sum_{n=0}^\infty \frac{s_n(x)}{n!} t^n
= \sum_{n=0}^\infty \frac{ns_{n-1}(x)}{n!}t^n = t \cdot g(t)e^{xF(t)}.
\een
This implies that the family $g(t)e^{xF(t)}$ is a family of coherent states indexed by $t$.

It is clear that all the discussion about the relationship between quantum mechanics
and associated sequences can be generalized to the Sheffer sequences.
We omit the details.

\subsubsection{Segal-Bargmann transform as a Sheffer operator}
Recall the exponential generating series of $He_n(x)$ is
\be
e^{xt-t^2/2} = \sum_{n=0}^\infty \frac{He_n(x)}{n!}t^n,
\ee
so by (c) above, $\{He_n(x)\}_{n\geq 0}$ is the Sheffer sequence for $(f(t),g(t)) = (t, e^{-t^2/2})$.
Therefore the map $He_n(x) \mapsto x^n$ and its inverse $x^n \mapsto He_n(x)$
are both Sheffer operators.

\section{Interpolating Statistics, Deformed Exponential Functions, and Deformed Entropy Functions}
\label{sec:Deformed-Entropy}

As we have discussed above,
the associated sequences in umbral calculus can be regarded as deformed power functions,
and hence their exponential generating series can be regarded as deformed exponential functions.
This leads us to the generalized thermostatics \cite{Naudts11} where another kind of deformed exponential functions
arise.

In this Section we will explain how the interpolating statistics are related
to deformed exponential functions and
deformed entropy functions.
We will show that all interpolating statistics
are deformed exponential functions in the sense of \cite{Naudts11}.

\subsection{The $\phi$-logarithm and $\phi$-exponential functions}
These were introduced by Naudt \cite{Naudts02}.
Fix a strictly positive non-decreasing function $\phi(u)$, defined on the positive
numbers $(0,+\infty)$.
The {\em $\phi$-logarithm function} is defined by:
\be
\ln_\phi(u) = \int_1^u dv \frac{1}{\phi(v)}, \;\;\; u > 0.
\ee
The inverse of the function $\ln_\phi(x)$ is called the
{\em $\phi$-exponential function} and is denoted $\exp_\phi(x)$.

A typical example used by Naudt \cite{Naudts02} is the case of $\phi(u)= u^q$.
In this case the deformed logarithm function is
\be
\ln_q(u) = \int_1^u \frac{1}{v^q} dv = \begin{cases}
\frac{u^{1-q}-1}{1-q}, & \text{if $q\neq 1$}, \\
\log(u), & \text{if $q=1$}.
\end{cases}
\ee
This was introduced in the context of nonextensive statistical physics
in Tsallis \cite{Tsallis88}.

\subsection{All interpolating statistics are $\phi_\bT$-exponential functions}

For our purpose we will take:
\ben
&& \phi_\bT(p) = p - \sum_{n \geq 2} T_{n-1} p^n,
\een
and perform the following formal calculation:
\ben
\widetilde{\ln}_{\phi_\bT}(p): & = & \int^p \frac{du}{\phi(u)}
= \int^p   \frac{du}{u - \sum_{n \geq 2} T_{n-1} u^n} \\
& = & \int^p \sum_{m=0}^\infty \biggl(\sum_{i=1}^\infty T_iu^i \biggr)^m \frac{du}{u} \\
& = & \int^p \sum_{n=0}^\infty \sum_{\sum_i m_ii =n}
\binom{m_1 + \cdots + m_n}{m_1, \dots, m_n} \prod_i T_i^{m_i}\cdot u^{n-1}du \\
& = & \ln(p) +  \sum_{n=1}^\infty \sum_{\sum_i m_ii =n}
\binom{m_1 + \cdots + m_n}{m_1, \dots, m_n} \prod_i T_i^{m_i}\cdot \frac{p^n}{n}.
\een
In other words,
\be \label{eqn:ln-phi}
\widetilde{\ln}_{\phi_\bT}(p) = \ln(p) + \sum_{n \geq 1} a_n(\bT) \frac{p^n}{n},
\ee
where each
\be \label{eqn:a-in-T}
a_n(\bT) = \sum_{\sum_i m_ii =n}
\binom{m_1 + \cdots + m_n}{m_1, \dots, m_n} \prod_i T_i^{m_i}
\ee
is a weighted homogeneous polynomial in $T_1, \dots, T_n$ of degree $n$,
with $\deg T_j = j$.
For example,
\ben
&& a_1 = T_1, \\
&& a_2 = T_2+T_1^2, \\
&& a_3 = T_3+T_2T_1 + T_1^3, \\
&& a_4 = T_4+2T_3T_1+T_2^2+3T_2T_1^2+T_1^4.
\een
Therefore, \eqref{eqn:ln-phi} shows that $\widetilde{\ln}_{\phi_\bT}$ is a formal deformation
of the $\ln$-function parameterized by $\{T_n\}$.

To find an expression for the formal inverse $\widetilde{\exp}_{\phi_\bT}$,
introduce the following series in $p$:
\be
X_{\phi_\bT}(p): = \exp(\tilde{\ln}_\phi(p)) = p (1 + \sum_{n\geq 1} b_n(\bT) p^n ),
\ee
where each $b_n(\bT)$ is a weighted homogeneous polynomial in $a_1, \dots, a_n$,
hence it is a weighted homogeneous polynomial in $T_1, \dots, T_n$:
\be \label{eqn:b-in-a}
b_n = \sum_{\sum_i m_ii=n} \prod_i \frac{a_i^{m_i}}{i^{m_i}m_i!}.
\ee
For example,
\ben
&& b_1 = a_1 = T_1, \\
&& b_2 = \frac{a_2}{2}+\frac{a_1^2}{2!}=\frac{T_2}{2}+T_1^2, \\
&& b_3 = \frac{a_3}{3}+ a_1\frac{a_2}{2}+\frac{a_1^3}{6}
=\frac{T_3}{3}+\frac{7}{6}T_2T_1 + T_1^3.
\een
Now applying the Lagrange inversion,
one can express $p$ as a formal power series in $X_{\Phi_T}$:
\be
p = X_{\phi_T}+\sum_{n=2}^\infty c_n(\bT) X_{\phi_T}^n,
\ee
for example,
\ben
&&c_2 = -b_1 = -a_1 = -T_1, \\
&&c_2 = -b_2+2b_1^2 = -\frac{1}{2}a_2+\frac{3}{2}a_1^2 = -\frac{1}{2}T_2+T_1^2, \\
&&c_3 = -b_3+5b_2b_1-5b_1^3
= -\frac{1}{3}a_3+2a_2a_1-\frac{8}{3}a_1^3
= -\frac{1}{3}T_3+\frac{4}{3}T_2T_1-T_1^3.
\een
So we get:
\be
\widetilde{\exp}_{\phi_\bT}(u) = e^u+\sum_{n=2}^\infty c_n(\bT) e^{nu}.
\ee
It is clearly a deformation of the ordinary exponential function parameterized by
$\{T_n\}$.

factor((1/24)*Chi([1,1,1,1], [4])+(1/6)*Chi([1,1,1,1], [1, 3])+(1/4)*Chi([1,1,1,1], [1, 1, 2])+(1/8)*Chi([1,1,1,1], [2, 2])+(1/24)*Chi([1,1,1,1], [1, 1, 1, 1]))

Now define a space $\Phi$ by:
\bea
&& \Phi = \{ \phi_\bT(p) = p - \sum_{n \geq 2} T_n p^n \}.
\eea
Recall the space of interpolating statistics is defined by
\be
\cS = \{ w(X) = X + \sum_{n \geq 2} w_n X^n\}.
\ee
Given any formal series $\phi_\bT(p) = p - \sum_{n \geq 2} T_n p^n \in \Phi$,
define a formal series $w_{\phi_\bT}$ in $\cS$ as follows:
\be
w_{\phi_\bT} = X + \sum_{n \geq 2} c_n(\bT) X^n.
\ee
Then $\phi_T(p)\mapsto w_{\phi_\bT}(X)$ defines  a map $g: \Phi\to \cS$.

\begin{Proposition}
The map $g$ is a one-to-one map.
\end{Proposition}

\begin{proof}
Given a formal power series $w(X) = X + \sum_{n \geq 2} w_n X^n$
 in the space $\cS$ of interpolating series,
one can apply Lagrange inversion to get the inverse series
$$X(w) = w + \sum_{n \geq 2} \hat{w}_n w^n.$$
Take $b_n = \hat{w}_{n-1}$.
One can solve \eqref{eqn:b-in-a} to express $\{a_n\}$ in terms of $\{b_n\}$,
and one can solve \eqref{eqn:a-in-T} to express $\{T_n\}$ in terms of $\{a_n\}$.
Therefore,
one can express $\{T_n\}$ in terms of $\{b_n=\hat{w}_{n-1}\}$,
and hence also of $\{w_n\}$.
This gives us an inverse map $\cS \to \Phi$.
More explicitly,
\be \label{eqn:phi-in-X}
\phi(u) = \frac{u+\sum_{n \geq 2} \hat{w}_n u^n}{1+\sum_{n\geq 2} n\hat{w}_n u^{n-1}}
= \frac{X(u)}{\frac{dX(u)}{du}}
= \frac{1}{\frac{d}{du} \log X(u)}.
\ee
\end{proof}

\begin{Proposition}
The following formula for $\log_\phi(p)$ holds when $\phi$ is determined by \eqref{eqn:phi-in-X}:
\be
\log_\phi(p) = \log X(u) - \log X(1).
\ee

\end{Proposition}

\begin{proof}
\ben
\log_\phi(p) & = &\int_1^p \frac{du}{\phi(u)}
= \int_1^p \frac{\frac{d}{du}X(u)}{X(u)} du \\
& = & \int_{X(1)}^{X(p)} \frac{dX}{X} = \log X(u) - \log X(1).
\een
\end{proof}

\subsection{The deduced logarithm and the  $\phi$-entropy}

Under the condition that
\be
c:=\int_0^1 \frac{v}{\phi(v)}dv < + \infty,
\ee
Naudt \cite[\S 10.3]{Naudts11} introduces a function:
\be
\chi(u) = \biggl[ \int_0^{1/u}  \frac{v}{\phi(v)}dv \biggr]^{-1},
\ee
and he calls the  deformed logarithm $\ln_\chi$ associated with $\chi$
the {\em deduced logarithm}.

Naudts \cite[\S 11.1]{Naudts11} defines the $\phi$-entropy by:
\be
H_\phi(\bp) = \sum_{i=1}^n  p_i \ln_\chi(1/p_i),
\ee
where $\bp=p_1, \dots, p_n$, $p_i \geq 0$, $\sum_i p_i = 1$.
After a short calculation:
\be
H_\phi(\bp)
= - \sum_{i=1}^n  p_i \int_1^{p_i}
\frac{1}{u^2} \biggl[\int_0^{u} dv \frac{v}{\phi(v)} \biggr] d u.
\ee

\begin{Proposition}
The partial derivatives of $\phi$-entropy are given by negative the $\phi$-logarithms
up to the constant $c$:
\be \label{eqn:H-phi-grad}
\frac{\pd H_\phi(\bp)}{\pd p_i} = - \log_{\phi}(p_i)-c,
\ee
where $c$ is a constant defined by:
\be
c=\int_0^1\frac{v}{\phi(v)}dv.
\ee
\end{Proposition}

\begin{proof}
We have the following computations:
\ben
\frac{\pd H_\phi(\bp)}{\pd p_i} & = & - \int_1^{p_i}  \frac{1}{u^2}
\biggl[\int_0^{u}\frac{v}{\phi(v)} dv\biggr] d u
-  \frac{1}{p_i}
 \int^{p_i}_0  \frac{u}{\phi(u)}du  \\
& = & \int^{p_i}_1
\biggl[\int_0^{u} \frac{v}{\phi(v)}dv \biggr] d \frac{1}{u}
-  \frac{1}{p_i}
 \int^{p_i}_0 du \frac{u}{\phi(u)}   \\
& = &  -\int_0^{p_i} du \frac{1}{\phi(u)} - \int_0^1\frac{v}{\phi(v)}dv
= \ln_{\phi}(p_i)-c.
\een
\end{proof}

\subsection{All deformed entropy functions are $\phi_\bT$-entropy functions}

Let $\xi(u)$ be defined by
\be
\xi_\phi(u)(u): = \int_0^{u}  \frac{v}{\phi(v)}dv .
\ee
It is related to $\chi(u)$ as follows:
\be
\xi_\phi(u)(u) = \frac{1}{\chi(1/u)}.
\ee
Denote by $\xi_\bT(u)$  the series $\xi(u)$ when $\phi$ is the series $\phi_\bT$.
We have
\ben
\xi_\bT(u) & = &
\int_0^u \frac{v}{v - \sum_{n \geq 2} T_{n-1} v^n}dv \\
& = & \int^u \sum_{n=0}^\infty \sum_{\sum_i m_ii =n}
\binom{m_1 + \cdots + m_n}{m_1, \dots, m_n} \prod_i T_i^{m_i}\cdot v^{n}dv \\
& = & u +  \sum_{n=1}^\infty \sum_{\sum_i m_ii =n}
\binom{m_1 + \cdots + m_n}{m_1, \dots, m_n} \prod_i T_i^{m_i}\cdot \frac{u^{n+1}}{n+1}.
\een

\begin{Proposition} \label{prop:xi-F}
Suppose that $F(X) = \sum_{n=1}^\infty \frac{w_n}{n}X^n$, so that
$X=X(w)$ is obtained by Lagrange inversion from
$w(X) = X \frac{dF(X)}{dX}$, and $\phi(u)$ is determined from $X(u)$ by
\ben
\phi(u) = \frac{1}{\frac{d}{du} \log X(u)},
\een
then we have
\be
\xi_\phi(u) = F(X(u)).
\ee
\end{Proposition}

\begin{proof}
This is proved by a straightforward computation:
\ben
\xi_\phi(u)(u) & = & \int_0^{u}  \frac{v}{\phi(v)}dv
= \int_0^u v \cdot \frac{d \log X(v)}{dv} dv \\
& = & \int_0^u \frac{w}{X(w)}dX(w) = \int_0^{X(u)} \frac{w(X)}{X} dX
= \int_0^{X(u)} \frac{d F(X)}{dX} dX \\
& = & F(X(u)).
\een
\end{proof}

We formally define
\be
\tilde{H}_{\phi_\bT}(p): = \int^p \frac{\xi_\bT(u)}{u^2}du.
\ee
By a formal calculations,
we have:
\ben
\tilde{H}_{\phi_\bT}(p)
& = & - p
 (\ln(p) + \sum_{n=1}^\infty \sum_{\sum_i m_ii =n}
\binom{m_1 + \cdots + m_n}{m_1, \dots, m_n} \prod_i T_i^{m_i}\cdot
\frac{p^{n}} {n(n+1)} ).
\een

Define a space of deformed entropy by:
\bea
&& \cH = \{ H_{\bs}(p) = -p( \ln(p) + \sum_{n=1}^\infty s_n \frac{p^n}{n(n+1))} \}.
\eea
Then one can define a map $f: \Phi \to \cH$
by sending $\phi(\bT) = - \sum_{n \geq 2} T_{n-1} p^n$ to
$H_{\bs}(p) = -p( \ln(p) + \sum_{n=1}^\infty s_n \frac{p^n}{n(n+1)} )$,
where
\be \label{eqn:S-in-T1}
s_n = \sum_{\sum_i m_ii =n}
\binom{m_1 + \cdots + m_n}{m_1, \dots, m_n} \prod_i T_i^{m_i}, \;\;\; n \geq 1.
\ee

\begin{Proposition}
The map $f: \Phi \to \cH$ is a one-to-one correspondence.
\end{Proposition}

\begin{proof}
For any deformation of the standard Boltzmann-Gibbs-Shannon entropy
of the form
$H_{\bs}(p) = -p( \ln(p) + \sum_{n=1}^\infty s_n \frac{p^n}{n(n+1)} )$,
one needs to find a unique  sequence of weighted homogeneous polynomial
$T_n = T_n(s_1, \dots, s_n)$ of degree $n$,
where $\deg s_n = n$,
such that $H_\bs(p) = \widetilde{H}_{\phi_\bT}(p)$.
Indeed,
one needs to solve the sequence \eqref{eqn:S-in-T1} of equations.
These are equivalent to
\be \label{eqn:S-in-T}
1 + \sum_{n \geq 1} s_n p^n = \frac{1}{1 - \sum_{n \geq 1} T_n p^n}.
\ee
Its solution is clearly:
\be
1 - \sum_{n \geq 1} T_n p^n = \frac{1}{1 + \sum_{n \geq 1} s_n p^n}.
\ee
\end{proof}

\subsection{Principle of maximum entropy applied to $\phi$-entropy}

Let us formally extremize the modified $\phi$-entropy functions
\be
\tilde{H}_{\phi}(\bp)
=   - \sum_{i=1}^n p_i \int^{p_i}  \frac{1}{v^2}
\biggl[\int^{v} du \frac{u}{\phi(u)} \biggr] d v
\ee
under the constraints
\bea
&& p_1 + \cdots + p_n = 1, \\
&& p_1 E_1 + \cdots + p_n E_n = E,
\eea
by the method of Lagrange multiplier.
The solution is given by
\ben
\frac{\pd}{\pd p_i} \tilde{H}_{\phi}(\bp) =  a_1 + bE_i, \;\;\; i =1, \dots, n,
\een
for some constants $a_1$ and $b$.
By \eqref{eqn:H-phi-grad},
\be
-\log_\phi(p_i) = a+bE_i,
\ee
where $a=a_1+c$.
We  let $X_i = e^{- (a+bE_i)}$.
So the solution is given by
\be
p_i = \exp_{\phi}(-(a+bE_i)), \;\;\;\; i=1, \dots, n,
\ee
and equivalently,
\be
X_i =\exp(\ln_\phi(p_i)).
\ee

\subsection{Generalized statistics as critical points of $\phi_\bT$-entropy}

Now we have the following commutative diagrams of one-to-one correspondences:
\be \label{eqn:S-Phi-H}
\xymatrix{
\cS \ar[dr]_{h = f g} \ar[r]^{g}
                & \Phi \ar[d]^{f}  \\
                & \cH             }
\ee
As a corollary,
we have: Every generalized statistics can be obtained by applying the Principle
of Maximum Entropy to the $\phi_\bT$-entropy function
for a unique $\phi_T \in \Phi$,
hence also to a unique deformed entropy $H_\bs(p) \in \cH$.

\subsection{Generalized Boson-Fermion correspondence on the spaces $\Phi$ and $\cH$}

Recall we have defined an involution $\sigma:\cS \to \cS$ in \S \ref{sec:Dual-Stat}.
Now using the diagram \eqref{eqn:S-Phi-H},
the involution $\sigma$ induces an involution
$\tau: \Phi\to \Phi$ and an involution  $\rho: \cH \to \cH$,
so that the following diagrams commute:
\ben
\xymatrix{
& \cS  \ar[r]^{g} & \Phi \ar[d]^{f} \\
\cS   \ar[r]_{g} \ar[ru]^{\sigma}
                & \Phi \ar[d]_{f} \ar[ru]^\tau & \cH \\
                & \cH      \ar[ru]_\rho       }
\een
We call these involutions  the
generalized Boson-Fermion correspondence on the spaces $\Phi$ and $\cH$ respectively.

\subsection{Entropy as  $\phi$-entropy}

In \S \ref{sec:Entropy} we have defined the entropy of the one-particle partition function
as the negative of the Legendre transformation of the free energy:
\be
H(X) = F(X) - \log X \cdot X \frac{dF(X)}{dX}.
\ee
In the above we have also considered the $\phi$-entropy $H_{\phi}(p)$:
\be
H_{\phi}(p) = -p \int_1^p \frac{1}{u^2} \biggl[\int_0^u \frac{v}{\phi(v)}dv\biggr] du,
\ee
where
\be
\phi(w) = \frac{X}{\frac{dX}{dw}} = X \frac{dw}{dX},
\ee
\be
w = X\frac{dF}{dX}.
\ee
The following is the main result of this paper:

\begin{Theorem} \label{thm:Main}
The entropy function $H(X)$ is related to the $\phi$-entropy function $H_{\phi}(p)$ as follows:
\be
H(X) = H_{\phi}(w(X)) + [F(X(1))-\log X(1)] \cdot w(X),
\ee
\be \label{eqn:H-F}
H_{\phi}(p) = F(X(p))
- p \cdot \log X(p) - p [F(X(1))-\log X(1)].
\ee
\end{Theorem}

\begin{proof}
The proof is similar to the proof of \eqref{eqn:H-phi-grad}:
\ben
H_{\phi}(p) & = & p \int_1^p   F(X(u)) d\frac{1}{u} \\
& = & p \cdot \frac{F(X(u))}{u}  \biggl|_1^p
- p \int_1^p \frac{1}{u} dF(X(u)) \\
& = & F(X(p)) - p F(X(1))
- p\int_1^p \frac{1}{w} \frac{dF(X(w))}{dX(w)} \cdot dX(w) \\
& = & F(X(p)) - p F(X(1))
- p\int_{X(1)}^{X(p)} \frac{1}{X} \cdot dX \\
& = & F(X(p)) - p F(X(1))
- p \cdot \log X(p) + p \log X(1).
\een
\end{proof}

\begin{Corollary}
The following identity holds:
\be
\frac{d}{dp} H_{\phi}(p) = - \log X(p) - [F(X(1))-\log X(1)].
\ee
\end{Corollary}

\section{Conclusions and Prospects}

\label{sec:Conclusions}

In this paper we first establish a link between interpolating statistics with umbral calculus.
This link inspires us to reexamine umbral calculus from the point of view of mathematical physics.
An unexpected outcome is that this link also leads to a connection of interpolating statistics
with generalized statistical mechanics, this not only opens the door of applications
of generalized statistical mechanics to fractional quantum Hall effects,
but also  enables us to understand the generalized entropy functions as the Legendre
transformations of the free energy of interpolating statistics.

Originally we are led to this work by spectral curves in Eynard-Orantin topological recursions,
which we understand as the genus zero one-point functions in various Gromov-Witten type theory.
The theory of interpolating statistics are concerned with one-particle partition functions,
so this leads to our definition of the spectral curves associated with interpolating statistics.
The original motivation to use umbral calculus is to find more examples
of interpolating statistics.
Fortunately  some familiarity with Ramanujan's second notebook enables us to identify
their spectral curves in the setting of Eynard-Orantin topological recursions.
We have recorded some examples in Section \ref{sec:Abel} and Section \ref{sec:Gould-Special}.

As a result,
the role of umbral calculus has changed in several aspects.
First, in the beginning it is used as a technical tool to find more examples,
but soon it turns out the connections with interpolating statistics and generalized entropy
suggest to reexamine the umbral calculus from the point of view of quantum mechanics
and statistical physics. We hope this interaction with mathematical physics
will lead to new progresses in this important branch of combinatorics.

Secondly,
because for some examples  the spectral curves also arise in topological string theory,
so it is natural to expect to find deeper connections between the objects in this paper
and string theory.
A string theoretical interpretation of either interpolating statistics, or generalized
entropy, or umbral calculus, is of course very desirable.
We hope to address this in future work.

Thirdly, umbral calculus can be regarded as belonging to formal algebraic geometry,
e.g. to formal group laws.
The latter is related complex cobordism theory and Hirzebruch genera.
Hence through this chain of connections,
one sees that interpolating statistics are related to deep theories in algebraic topology.
This partly reflects that the topological nature of  fractional quantum Hall effect.
We will elaborate on this point in subsequent work.

We remark that it seems to be natural to use generalized statistical physics
to study interpolating statistics,
in particular the fractional quantum effect.
This is because fractional quantum Hall effect is an example of
topological orders.
Microscopically, topological orders correspond to patterns of long-range quantum entanglement.
On the other hand, as Tsallis remarked in \cite{Tsallis01}
 nonextensive statistical mechanics seem to be
more suitable to describe long-range interactions.

In the community of cybernetics and information  at least 25
different entropy functions  have been advanced
(see e.g. Taneja \cite{Taneja}).
These functions are introduced for various specific purposes, e.g., image processing.
Except for the Shannon entropy, the other 24 examples in {\em loc. cit.} are different from
the examples in the following Appendix.
It is not clear at present whether the examples in this paper are useful in that area.

Finally,
the examples in the Appendix lead us to connections with operads and cluster algebras.
We will report such connections in subsequent work.

\vspace{.2in}
{\bf Acknowledgements}.
The author is partly supported by NSFC grants 11661131005
and 11890662. The author thanks Professor Sen Hu and Professor Guowu Meng for
introducing him to fractional quantum Hall effects.
This is an expanded version of a manuscript that has a limited circulation.
Some of the results were reported in a colloquium talk at Peking University in 2018.
The author thanks Professor Hunjun Fan for the invitation and the audience for their interest.

\begin{appendix}

\section{Examples of Interpolating Statistics from  Umbral Calculus}

\label{sec:Examples}

In this Appendix we present some examples from umbral calculus.
Most of them can be found in Chapter 3 of Ramanujan's second notebook \cite{Ram},
or Chapter 4 of Roman's book on umbral calculus \cite{Roman}.
They are reproduced here because of the following two reasons.
On the one hand we want to provide some extra information about
their related spectral curves;
it is a surprise that some of them reproduce
the spectral curves that appear in the literature on Eynard-Orantin topological recursions.
On the other hand,
we will present the calculations related to $\phi$-logarithm, $\phi$-exponential, and $\phi$-entropy.
The motivation is to illustrate the relationship between
umbral calculus and generalized entropy
studied in \S \ref{sec:Deformed-Entropy}
by concrete examples.

We will also present some examples not in Roman \cite{Roman}.
One of them comes from the dilogarithm,
and three others are obtained by averaging the Acharya-Swamy
statistics.
In the former example,
a famous identity of Euler for dilogarithm naturally arises,
and the $\phi$-entropy is very close to the Rogers dilogarithm.
In the latter examples,
several integer sequences of rich combinatorial significance appear in the computations.
We also include the Gentile statistics and an example motivated by it.
Some of our examples can be identified with the examples
in the Appendix to Taylor \cite{Taylor} where some
some more interesting examples can be found.

Let us recall the notations.
If $\{\gamma_n(x)\}$ is a polynomial sequence of binomial type,
then
\be
\sum_{n=0}^\infty \gamma_n(x) \frac{X^n}{n!}
= \exp (x F(X))
\ee
for a formal power series $F(X) = X + O(X^2)$.
The relationship between the free energy $F(X)$ and $\{\gamma_n\}$ is given in the following
two identities:
\be \label{eqn:gamma-n}
\gamma_n(x) = n! \cdot \sum_{k=1}^n \frac{x^k}{k!} (F(X)^k)|_{X^n},
\ee
where $(\cdot)_{X^n}$ means the coefficient of $X^n$,
and
\be
F(X) = \sum_{n=0}^\infty \gamma_n'(0) \frac{X^n}{n!}.
\ee

Furthermore,
let $f(Y)$ be the compositional inverse series of $F(X)$,
i.e.,
\begin{align*}
f(F(X)) & = X, & F(f(Y)) & = Y.
\end{align*}
Then one has
\be
e^{xY} = \sum_{n=0}^\infty \gamma_n(x) \frac{f(Y)^n}{n!}
\ee
and the following recursion relations are satisfied:
\bea
&& f(D) \gamma_n(x) = n \cdot \gamma_{n-1}(x), \\
&& \gamma_n(x) = x [f'(D)]^{-1}\gamma_{n-1}(x), \;\;\; n \geq 1,
\eea
where  $D=d/dx$.

To find the series $\phi$, we first find:
\be
w(X) = X\frac{d}{dX}F(X),
\ee
then we apply Lagrange inversion to get:
\be
X = X(w).
\ee
Next apply \eqref{eqn:phi-in-X} to get $\phi$:
\be
\phi(u) = \frac{1}{\frac{d}{du} \log X(u)}
= \frac{X(u)}{\frac{d}{du} X(u)}.
\ee
The $\phi$-logarithm is computed from $\phi(u)$ by the following formula:
\be
\ln_{\phi}(p)  =  \int^p_1 \frac{du}{\phi(u)} ,
\ee
the result is
\be
\ln_{\phi}(p)   = \log \frac{X(p)}{X(1)}.
\ee
The $\phi$-exponential function is the inverse function of $\ln_{\phi}$.
To get the $\varphi$-entropy,
we first compute
\be
\xi_{\phi}(u) = \int^u_0 \frac{vdv}{\phi(v)},
\ee
the result is
\be
\xi_{\phi}(u) = F(X(u)).
\ee
Next, the $\phi$-entropy can be computed as follows:
\be
H_{\phi}(p) = -p \cdot \int_1^p \frac{\xi(u)}{u^2}du
= -p \cdot \int_1^p \frac{F(X(u))}{u^2}du.
\ee
The result is
\be
H_{\phi}(p) = F(X(p))
- p \cdot \log X(p) - p [F(X(1))-\log X(1)].
\ee

\subsection{The Boltzmann-Gibbs statistics}
In this case,
\begin{align*}
\gamma_k^{BG}(n) & = x^n, &
F^{BG}(X) & = X, &
w^{BG}(X) & = X, &
f^{BG}(Y) & = Y.
\end{align*}
The spectral curves are given by:
\begin{align*}
z & = e^X, &
X & = Y.
\end{align*}
For this example,
\begin{align*}
w^{BG}(X) & = X & X^{BG}(w) & = w,
\end{align*}
it follows that
\begin{align*}
\phi^{BG}(p) & = p, & \xi^{BG}(u) & = u.
\end{align*}
The corresponding entropy is given by:
\ben
H^{BG}(p) = -p \int_1^p \frac{1}{u^2} \cdot u du = -p \log p.
\een

\subsection{The  Fermi-Dirac statistics}
In this case we have
\ben
&& z^x = (1+X)^x, \\
&& \gamma_n^{FD}(x) = (x)_n = x(x-1)(x-2) \cdots (x-(n-1)), \\
&& F^{FD}(X) = \ln (1+X) = \sum_{n=1}^\infty \frac{(-1)^{n-1}}{n} X^n, \\
&& f^{FD}(Y) = e^Y - 1.
\een
The spectral curves are given by:
\begin{align*}
z & = 1+ X, &
X & = e^Y - 1.
\end{align*}
The identity $z^x \cdot z^y = z^{x+y}$ is equivalent to
the Chu-Vandermonde identity:
\be
\sum_{i+j=n} \binom{x}{i} \binom{y}{j} = \binom{x+y}{n}.
\ee

In this case we have
\begin{align*}
w^{FD}(X) & = \frac{X}{1+X}, &
X^{FD}(w) = \frac{w}{1-w}.
\end{align*}
From this we compute that:
\ben
\phi^{FD}(u) & = & \frac{1}{\frac{d}{du} \log X^{FD}(u)} = u(1-u),  \\
\xi^{FD}(u) & = &  \int^u \frac{v}{v(1-v)}dv = -\log (1-u)
= F^{FD}(X^{FD}(u)).
\een
The $\phi^{FD}$-logarithm is
\ben
\ln_{\phi^{FD}}(p) = \int^p \frac{1}{u(1-u)}du =  \log \frac{p}{1-p}
= \log X^{FD}(p).
\een
The $\phi^{FD}$-exponential is given by:
\be
p = \frac{e^X}{e^X+1} = \frac{1}{1+e^{-X}}.
\ee
The $\phi^{FD}$-entropy is:
\ben
H^{FD}(p)
& = &  p \int_1^p \frac{1}{u^2} \log (1-u) du \\
& = & -p \log p - (1-p) \log(1-p).
\een

\subsection{The Bose-Einstein statistics}
In this case we have:
\ben
&& z^x = \frac{1}{(1-X)^x}, \\
&& \gamma_n^{BE}(x) = x^{(n)} = x(x+1)(x+2) \cdots (x+(n-1)), \\
&& F^{BE}(X) = -\ln (1-X) = \sum_{n=1}^\infty \frac{1}{n} X^n, \\
&& f^{BE}(Y) = 1- e^{-Y}.
\een
The spectral curves are given by:
\begin{align*}
z & = \frac{1}{1-X}, &
X & = 1 - e^{-Y}.
\end{align*}
The identity $z^x \cdot z^y = z^{x+y}$ is equivalent to
the Chu-Vandermonde identities with $x, y$ changed to $-x, -y$ respectively:
\be
\sum_{i+j=n} (-1)^i \binom{-x}{i}\cdot (-1)^j \binom{-y}{j} = (-1)^n\binom{-(x+y)}{n}.
\ee

In this case we have:
\begin{align*}
w^{BE}(X) & = \frac{X}{1-X}, &
X^{BE}(w^{BE}) & = \frac{w^{BE}}{1+w^{BE}}.
\end{align*}
So we have
\begin{align*}
\phi^{BE}(p) & = p(1+p), & \xi^{BE}(u) = \log (1+u).
\end{align*}
The $\phi^{BE}$-logarithm is
\ben
\ln_{\phi^{BE}}(p) = \int^p \frac{1}{u(1+u)}du =  \log \frac{p}{1+p}= \log X^{BE}(p).
\een
The $\phi^{BE}$-exponential is given by:
\be
p = \frac{e^X}{1- e^X} = \frac{1}{e^{-X}-1}.
\ee
The $\phi^{BE}$-entropy is:
\ben
H^{BE}(p)
& = &  -p \int_1^p \frac{1}{v^2} \log (1+v) dv \\
& = & -p \log p + (1+p) \log(1+p)-2\log 2.
\een

\subsection{The Acharya-Swamy statistics}
In this case one has:
\ben
&& \gamma_n^{AS}(x; \epsilon) = x(x+\epsilon)(x+2\epsilon)
\cdots (x+(n-1)\epsilon), \\
&& F^{AS}(X; \epsilon) = \frac{1}{\epsilon}\ln (1 + \epsilon X)
= \sum_{n=1}^\infty \frac{(-1)^{n-1}}{n} \epsilon^{n-1} X^n, \\
&& f^{AS}(Y; \epsilon) = \frac{e^{\epsilon Y}-1}{\epsilon}.
\een
The spectral curves are given by:
\begin{align*}
z & = (1+\epsilon X)^{1/\epsilon}, &
X & = \frac{e^{\epsilon Y}-1}{\epsilon}.
\end{align*}
By taking the logarithmic derivative of $F^{AS}$ we get:
\begin{align*}
w^{AS}(X;\epsilon) & = \frac{X}{1+\epsilon X},
\end{align*}
By taking the Lagrange inversion one gets:
\begin{align*}
X^{AS} &= \frac{w^{AS}}{1-\epsilon w^{AS}}.
\end{align*}
A combinatorial interpretation of this by counting trees is given by Parker's Theorem
(see e.g. Gessel \cite{Gessel}).

 From these we have the following computations:
\begin{align*}
\phi^{AS}(p) & = \frac{1}{\frac{d}{dp} \log X_{\phi^{AS}}(p)}
= p(1-\epsilon p), &
\xi^{AS}(u) & = - \frac{1}{\epsilon} \log (1-\epsilon u).
\end{align*}
\ben
\ln_{\phi^{AS}}(p)& = & \int^p \frac{1}{u(1-\epsilon u)}du =  \log \frac{p}{1- \epsilon p} =\log X^{AS}(p), \\
p & = & \frac{e^X}{1+\epsilon e^X} = \frac{1}{e^{-X}+\epsilon}, \\
H^{AS}(p) & = &  \frac{p}{\epsilon} \int_1^p \frac{1}{v^2} \log (1-\epsilon v) dv \\
& = & -p \log p - \frac{1}{\epsilon}(1-\epsilon p) \log(1-\epsilon p)
+ \frac{1}{\epsilon}(1-\epsilon) \log(1-\epsilon).
\een

\subsection{The  Gentile statistics}

Historically the first statistics that was proposed to interpolate between the Bose-Einstein
and the Fermi-Dirac statistics is the Gentile statistics.
It is a family of  one-particle partition functions:
\be
z_{(p)}(x) = 1+ X + X^2 + \cdots + X^p,
\ee
with the following property:
\begin{align}
z_{(1)}(X) & = z^{FD}(X), &
\lim_{p\to +\infty} z_{(p)}(X) & = z^{BE}(X).
\end{align}
One can compute $w_n = p_n(\bX)$ as follows:
\ben
F & = & \log(1+ X + \cdots + X^p)
= \log \frac{1-X^{p+1}}{1-X} \\
& = & \sum_{n=1}^\infty \frac{1}{n} (X^n - X^{(p+1)n}),
\een
and so
\ben
w(X) & = & \sum_{n=1}^\infty (X^n - (p+1)X^{(p+1)n})
= \frac{X}{1-X} - (p+1) \frac{X^{p+1}}{1-X^{p+1}}.
\een

\subsection{An example motivated by Gentile statistics}
\label{sec:Gentile-Mot}
We take
\be
F(X) = X + X^2 + \cdots = \frac{X}{1-X}.
\ee
Its inverse series is
\be
f(X) = X - X^2 + \cdots = \frac{X}{1+X}.
\ee
The exponential generating series of the conjugate series of $F(X)$ is
\be
\sum_{n=0}^\infty \gamma_n(x) \frac{X^n}{n!}
= \exp \frac{xX}{1-X}.
\ee
The following are the first few examples of $\gamma_n(X)$:
\ben
&& \gamma_1(x) = x, \\
&& \gamma_2(x) = x^2+2x, \\
&& \gamma_3(x) = x^3+6x^2+6x, \\
&& \gamma_4(x) = x^4 + 12x^3 + 36 x^2+ 24x,
\een
for the combinatorial meaning of the coefficients, see A089231,
A111596, A066667, A008297 on \cite{Sloane}.
They are related to the Lah numbers and the Laguerre polynomials
of order $-1$.
See Taylor's thesis \cite[Section 3.4 and p. 97]{Taylor} for its relation  to
counting permutations.

For this example,
we have
\ben
&& w(X) = \frac{X}{(1-X)^2},
\een
and so
\ben
X & = & \frac{1+2w-\sqrt{1+4w}}{2w} \\
& = & w-2 w^2+5 w^3-14 w^4+42 w^5-132 w^6+429 w^7-1430 w^8+\cdots,
\een
where the coefficients of the series expansion are the Catalan numbers up to signs.
\ben
\phi(w) & = &  X \frac{dw}{dX} = \frac{X(1+X)}{(1-X)^3}
= w \sqrt{1+4w} \\
& = & w+2 w^2-2 w^3+4 w^4-10 w^5+28 w^6-84 w^7+\cdots.
\een
Up to signs, the coefficients are A002420 on \cite{Sloane}.
\ben
\xi_\phi(w) & = & \int_0^w \frac{w}{\phi(w)} dw
= \int^w_0 \frac{1}{\sqrt{1+4w}}dw
= \int_0^X \frac{w}{X} dX \\
& = & \frac{\sqrt{1+4w}-1}{2} =F(X).
\een
\ben
H_\phi(p) & = & -p \int_1^p \frac{\xi_\phi(w)}{w^2} dw
= -p \int^w_1 \frac{\sqrt{1+4w}-1}{2w^2}dw \\
& = & \frac{1-2p-\sqrt{1+4p}}{2}  - p \log\frac{1+2p-\sqrt{1+4p}}{2p}
+ cp
\een
for some constant $p$.

\subsection{The exponential polynomials}
\label{sec:Exponential}

For this sequence one has:
\ben
&& F(X) = e^X-1, \\
&& f(Y) = \log (1+Y),
\een
a combinatorial interpretation by counting trees is given by Drake's Theorem \cite{Drake}
(see also Gessel \cite{Gessel}).
The sequence is given by
\ben
&& \gamma_n(x) =\sum_{k=0}^n S(n,k) x^k,
\een
where $S(n,k)$ are the Stirling numbers of the second kind:
\be
S(n,k) = \frac{1}{k!} \sum_{j=0}^k (-1)^{k-j} \binom{k}{j} j^n.
\ee
The spectral curves are given by:
\begin{align*}
z & = e^{e^X-1}, &
Y & = e^X-1, & X & = \log(1+Y) .
\end{align*}
In this case,
\ben
w & = X\frac{dF}{dX} = Xe^X,
\een
By Lagrange inversion formula:
\ben
&& X = \sum_{n=1}^\infty \frac{(-1)^{n-1}n^{n-1}}{n!} w^n.
\een
This is essentially Cayley's famous formula for counting rooted trees.

Some other relevant computations are:
\ben
\phi(w) & = & \frac{1}{\frac{d}{dw} \log X}
= X \frac{dw}{dX} = Xe^X +X^2e^X \\
& = & w\biggl(1+ \sum_{n=1}^\infty \frac{(-1)^{n-1}n^{n-1}}{n!} w^n\biggr).
\een

\ben
\xi_\phi(w) & = & \int_0^w \frac{w}{\phi(w)} dw
= \int^w_0 \frac{w}{X\frac{dw}{dX}}dw
= \int_0^X \frac{w}{X} dX \\
& = & \int^X_0 e^X dX = e^X - 1 =F(X) \\
& = & \sum_{n=1}^\infty (-1)^{n-1}\frac{(n-1)^{n-1}}{n!}w^n.
\een

\ben
\log X = \log w - X
\een
\ben
H(p)& = & -p\int_0^p \frac{1}{w^2} (e^{X(w)}-1)dw \\
& = & - p \log p
+ \sum_{n=2}^\infty (-1)^{n}
\frac{(n-1)^{n-2}}{n!}p^n.
\een

\subsection{Abel polynomials and Lambert series}
\label{sec:Abel}

In 1826 Abel   proved the following deep generalization
of the binomial identity (cf. \cite[\S 3.1]{Comtet}):
\be
(x+y)^n = \sum_{k=0}^n \binom{n}{k} x(x-k a)^{k-1}(y+k a)^{n-k}.
\ee
This leads to the sequence of Abel polynomials:
\ben
&& \gamma_n^{Abel}(x) = x(x-na)^{n-1}, \\
&& f^{Abel}(Y) = Y e^{aY}, \\
&& F^{Abel}(X) = \sum_{n=1}^\infty \frac{(-an)^{n-1}}{n!} X^n, \\
&& w^{Abel}(X) = \sum_{n=1}^\infty \frac{(-an)^{n-1}}{(n-1)!} X^n.
\een
This example is related to Cayley's formula for rooted trees and
Entry 13 of Ramanujan's second notebook \cite{Ram}.
The spectral curves are:
\begin{align*}
z & = \exp \sum_{n=1}^\infty \frac{(-an)^{n-1}}{n!} X^n, &
X & = Y e^{aY}.
\end{align*}
Spectral curves of this form
coincide exactly with the spectral curves in the studies
of Eynard-Orantin topological recursions of Hurwitz numbers \cite{Bou-Mar, Eyn-Mul-Saf}
and its generalizations \cite{Bou et al}.

\ben
X & = & w^{Abel} \cdot \exp \sum_{n=1}^\infty \frac{n+1}{n}a^n (w^{Abel})^n \\
& = & \frac{w^{Abel}}{1-aw^{Abel}}\exp \biggl(\frac{1}{1-aw^{Abel}}-1\biggr).
\een

\ben
\phi^{Abel}(p) & =& \frac{p}{1 + \sum_{n=1}^\infty (n+1) a^n p^n} = p(1-ap)^2,\\
\xi^{Abel}(u) & = & \int_0^u
\frac{v}{v(1-a v)^2}dv =\frac{1}{a(1-a u)}
-\frac{1}{a} = \frac{u}{1-au}, \\
\log^{Abel}(p) & = & \int_1^p \frac{1}{v(1-av)^2} dv = \log \frac{p}{1-ap}+\frac{1}{1-ap},\\
H^{Abel}(p)
& = &  -p
\int_1^p \frac{1}{u^2} \biggl[ \frac{1}{a(1-a u)}
-\frac{1}{a}\biggr] du \\
& = & -p \log p + p \log(1-a p)-p\log(1-a).
\een

For $e^{xf^{Abel}(X)}=\sum_{n=0}^\infty \frac{\gamma^\vee(X)}{n!}X^n$,
see Taylor \cite[\S 3.3 and p. 96]{Taylor}.
They are related to A059297 on \cite{Sloane}.

\subsection{Gould polynomials}

After a change of variables from the original definition in Roman \cite{Roman},
Gould polynomials are defined to be
the following sequence of polynomials:
\be \label{Gould}
\gamma_n^{Gould}(x;a, b) = x \prod_{j=1}^{n-1} (x-an-jb).
\ee
They form   the sequence of polynomials of binomial type
associated with
\be
f^{Gould}(Y; a, b) = e^{aY}\frac{e^{bY} -1}{b}, \;\;\; b \neq 0,
\ee
whose compositional inverse series is
\be
F^{Gould}(X; a, b) = \sum_{k=1}^\infty (-1)^{k-1}
\prod_{j=1}^{k-1} (ak+jb) \cdot \frac{X^k}{k!}.
\ee
In particular,
\ben
&& \sum_{n=0}^\infty x \prod_{j=1}^{n-1} (x-an-bj) \cdot X^n \\
& = & \exp \biggl(x
\sum_{k=1}^\infty (-1)^{k-1}
\prod_{j=1}^{k-1} (ak+jb)
\cdot \frac{X^k}{k!} \biggr).
\een
The spectral curves are
\be
z = \exp  \sum_{k=1}^\infty (-1)^{k-1}
\prod_{j=1}^{k-1} (ak+jb) \cdot \frac{X^k}{k!}
\ee
and
\be
X = e^{aY}\frac{e^{bY} -1}{b}.
\ee
The latter can be rewritten as follows:
\be \label{eqn:GouldSpec}
b X e^{-aY} - e^{bY} + 1 = 0.
\ee
One can compare it with Entry 14 in Chapter 3 of Ramanujan's second notebook
\cite[(14.4)]{Ram}.
One also has:
\be
w^{Gould}(X; a, b) = \sum_{k=1}^\infty (-1)^{k-1}
k \prod_{j=1}^{k-1} (ak+jb) \cdot \frac{X^k}{k!}.
\ee
It satisfies the following equation:
\be
X = \frac{w (1- aw)^{a/b}}{(1-(a+b)w)^{(a+b)/b}}.
\ee
This generalizes Wu's formula ($a=\alpha$, $b=-1$):
\ben
X = \frac{w}{(1- \alpha w)^\alpha [1 + (1-\alpha) w]^{1-\alpha}}
\een
for Haldane-Wu statistics \cite{Haldane, Wu}.
We then have:
\ben
\phi^{Gould}(p) & = & p(1- a p)(1-(a+b)p), \\
\xi^{Gould}(u) & = & \int_0^u \frac{v}{v(1-av)(1-(a+b)v)}dv \\
& = & \frac{1}{b} (\log(1-au)-\log(1-(a+b)u)), \\
H^{Gould}(p) & = & -p \int^p \frac{1}{u^2} \cdot \frac{1}{b} (\log(1-au)-\log(1-(a+b)u))du \\
& = & -p\log p+ \frac{1}{b}\biggl[
(1-ap)\log (1-ap) \\
&& -(1-(a+b)p)\log(1-(a+b)p) \biggr].
\een

\subsection{Specializations of the spectral curve associated with Gould polynomials}
\label{sec:Gould-Special}

Let us consider the following specializations of the spectral curve \eqref{eqn:GouldSpec}.

Case 1. By taking $a=0$, $b = \epsilon$,
one gets
\be
X = \frac{e^{\epsilon Y}-1}{\epsilon},
\ee
hence one recovers the spectral curve associated with the Acharya-Swamy statistics.

Case 2. By taking $b \to 0$, one gets
\be
X = Y e^{aY},
\ee
and one recovers the Lambert curve and its generalizations mentioned in \S \ref{sec:Abel}.

Case 3. When $a=g-1$, $b=1$,
one gets the curve:
\be
X = e^{gY} - e^{(g-1)Y}.
\ee
This is the spectral curve for the framed topological vertex \cite{Bou-Mar, Chen, Zhou}.

Case 4. When $b=-2a$, one gets
\be
-2aX = e^{-aY} - e^{aY}.
\ee
After suitable change of coordinates,
one gets the Catalan curve:
\be
z= w - \frac{1}{w}
\ee
that appears in Eynard-Orantin topological recursions of several geometric problems.

\subsection{Mittag-Leffler polynomials}

For this sequence one has:
\ben
f^{ML}(X) & = & \frac{e^X-1}{e^X+1}, \\
F^{ML}(X) & = & \log\biggl( \frac{1+X}{1-X}\biggr), \\
\gamma^{ML}_n(x) & = & \sum_{k=0}^n \binom{n}{k}\binom{n-1}{n-k} 2^k (x)_k.
\een
The exponential generating series of Mittag-Leffler polynomials is
\be
\sum_{n=0}^\infty \frac{\gamma_n(x)}{n!}t^n = \biggl(\frac{1+t}{1-t} \biggr)^x
\ee
By taking $t=X$ and $x=1$, one sees that the spectral curves are:
\begin{align*}
z & = \frac{1+X}{1-X}, &
Y & = \log\biggl(\frac{1+X}{1-X}\biggr), &
X & = \frac{e^Y-1}{e^Y+1}.
\end{align*}

For simplification of notations we take
$$F^{ML}(X)=\log \biggl(\frac{1+X/2}{1-X/2}\biggr).$$
Then we have
\ben
w^{ML} & = &\frac{X}{1-\frac{X^2}{4}},
\een
and so
\ben
X^{ML}(w) & = &\frac{2(\sqrt{1+w^2}-1)}{w}
= 2\sum_{n=1}^\infty (-1)^{n-1} \frac{(2n-2)!}{n!(n-1)!}
(w/2)^{2n-1},
\een
where the coefficients $\frac{(2n-2)!}{n!(n-1)!}$
are the Catalan numbers,
so $X^{ML}$ is essentially a generating function
of the Catalan numbers.
\ben
\phi^{ML}(p)& = & \frac{1}{p\frac{d}{dp} \log X^{ML}(p)} =p\sqrt{p^2+1} \\
& = &
p \biggl( 1 + 2 \sum_{n=1}^\infty (-1)^{n-1}
\frac{(2n-2)!}{n!(n-1)!} (p/2)^{2n} \biggr).
\een

\ben
\xi^{ML}(u) & = & \int_0^u \frac{v}{\phi^{ML}(v)}dv
= \int_0^u \frac{dv}{\sqrt{v^2+1}}
= \log (u+\sqrt{1+u^2}) \\
& = & 2\sum_{n=1}^\infty (-1)^{n-1} \frac{(2n-2)!}{(2n-1)\cdot(n-1)!^2} (u/2)^{2n-1}.
\een

\ben
\log_{\phi^{ML}}(p)
& = & \log X^{ML}(p) = \log \biggl( \frac{2(\sqrt{1+p^2}-1)}{p}\biggr)\\
& = & \log p
+ \sum_{n=1}^\infty \frac{(-1)^n}{2^{2n}} \frac{\binom{2n}{n}}{2n} p^{2n}.
\een

\ben
H^{ML}(p)
& = & - p \int^p \frac{\xi^{ML}(u)}{u^2}du = -p\int^p \frac{\log (u+\sqrt{1+u^2})}{u^2}du \\
& = & p \biggl(-\log p + \log\frac{(1+\sqrt{1+p^2})}{2}+\frac{\log (p+\sqrt{1+p^2})}{p}-1\biggr) \\
&= & -p\log p +  \sum_{n=2}^\infty \frac{(-1)^{n-1}}{2^{2n-2}}
 \frac{(2n-2)!}{(2n-1)(2n-2)\cdot(n-1)!^2} u^{2n-2}.
\een

\subsection{Bessel polynomials}
For this sequence,
\ben
&& f^{Bessel}(t) = t - t^2/2, \\
&& F^{Bessel}(t) = 1- \sqrt{1-2t}, \\
&& \gamma_n^{Bessel}(X) = \sum_{k=1}^n \frac{(2n-k-1)!}{(k-1)!(n-k)!}\biggl(\frac{1}{2}\biggr)^{n-k}x^k.
\een
The exponential generating function is
\ben
&& \sum_{n=0}^\infty \frac{\gamma_n(x)}{n!}t^n = e^{x(1-(1-2t)^{1/2})},
\een
and so the spectral curves are
\begin{align*}
z& = e^{1-(1-2X)^{1/2}}, & Y & = 1- (1-2X)^{1/2}, & X & = Y-\frac{Y^2}{2}.
\end{align*}
We have:
\ben
w^{Bessel} & = & \frac{X}{\sqrt{1-2X}}, \\
X^{Bessel}(w) & = & w (\sqrt{1+w^2}-w), \\
\phi^{Bessel}(p) & = & p\sqrt{1+p^2}(\sqrt{1+p^2}+p).
\een
\ben
\xi^{Bessel}(u)
& = & \int_1^v \frac{v}{v\sqrt{1+v^2}(\sqrt{1+v^2}+v)}dv \\
& = & u + 1-\sqrt{1+u^2}.
\een
One again encounters Catalan numbers in their expansions.
We also have:
\ben
\log_{\phi^{Bessel}}(p)
& = & \log X(p) = \log[p (\sqrt{1+p^2}-p)] \\
& = & \log p -
\sum_{n=0}^\infty \frac{(-1)^n}{2^{2n}}
\frac{\binom{2n}{n}}{2n+1} p^{2n+1}.
\een

\ben
H^{Bessel}(p)
& = & -p\int^p \frac{u + 1-\sqrt{1+u^2}}{u^2}du \\
&= & -p\log p -\sqrt{1+p^2}+ p\log(p+\sqrt{1+p^2})+1 \\
& = & -p \log p +
\sum_{n=0}^\infty \frac{(-1)^n}{2^{2n}}
\frac{\binom{2n}{n}}{(2n+1)(2n+2)} p^{2n+2}.
\een

For $e^{xf^{Bessel}(X)}=\sum_{n=0}^\infty \frac{\gamma^\vee(X)}{n!}X^n$,
see Taylor \cite[\S 3.2 and p. 96]{Taylor}.

\subsection{Mott polynomials}

The Mott polynomials (cf. \cite[p. 251]{Erdelyi} and \cite[\S 4.12]{Roman}) are defined by:
\be \label{eqn:Mott}
\exp \biggl[ x \cdot \frac{(1-t^2)^{1/2}-1}{t}\biggr]
= \sum_{n=0}^\infty  g_n(x) t^n,
\ee
where $g_n(x)$ are explicitly given by:
\ben
&& g_n(x) = (-x/2)^n(n-1)!
\sum_{l=0}^{[n/2]} \frac{x^{-2l}}{l!(n-l)!(n-2l-1)!}.
\een
hence they are the associated polynomials
with
\ben
&& F (t) = \frac{(1-t^2)^{1/2}-1}{t},
\een
and so
\ben
&& f(t) = \frac{-2t}{1+t^2}.
\een
In \eqref{eqn:Mott}, change $t$ to $-2t$,
\be \label{eqn:Mott2}
\exp \biggl[ x \cdot \frac{1-(1-4t^2)^{1/2}}{2t}\biggr]
= \sum_{n=0}^\infty  \frac{\gamma_n(x)}{n!} t^n,
\ee
where for $n \geq 0$,
\ben
&& \gamma_n(x) = n!(n-1)!
\sum_{l=0}^{[n/2]} \frac{x^{n-2l}}{l!(n-l)!(n-2l-1)!}
\een
form the associated sequence with
\ben
&& F^{Mott}(t) = \frac{1-(1-4t^2)^{1/2}}{2t}, \\
&& f^{Mott}(t) = \frac{t}{1+t^2}.
\een
The spectral curves are:
\begin{align*}
z &= \exp \biggl[ \frac{1-(1-4X^2)^{1/2}}{2X}\biggr], &
Y & = \frac{1-(1-4X^2)^{1/2}}{2X}, &
X & = \frac{Y}{1+Y^2}.
\end{align*}

\be \label{eqn:w-Mott}
w^{Mott}(X)  =  \frac{1-\sqrt{1-4X^2}}{2X\sqrt{1-4X^2}}
= \sum_{n=1}^\infty \binom{2n-1}{n} X^{2n-1}.
\ee
By Cardano formula,
\ben
X^{Mott}(w) & = & (\frac{1}{6w}(1+18w^2+3^{3/2}w\sqrt{1+11w^2-w^4})^{1/3} \\
& + & \frac{1+3w^2}{6w}(1+18w^2+3^{3/2}w\sqrt{1+11w^2-w^4})^{-1/3}-\frac{1}{3w}.
\een
The first few terms of the series expansion are:
\ben
X^{Mott}(w)=w-3w^3+17w^5-119w^7+929w^9-7755w^{11}+\cdots
\een
The coefficients $1,-3,17,-119,929, -7755, \dots$ do not form a sequence on \cite{Sloane}.
The $\phi$-function in this case
\ben
\phi^{Mott}(w) = \frac{1}{\frac{d}{dw}\log X^{Mott}(w)}
\een
has a complicated explicit expression.
The first terms of its series expansion are:
\ben
\phi^{Mott}(w) & = & w+6w^3-14w^5 +78w^7-542w^9+4214w^{11}+\cdots.
\een
The sequence of coefficients $1,6,14,78,542,4214,\dots$ does not appear on \cite{Sloane}.
The expression of $\phi^{Mott}$ in terms of $X=X^{Mott}$ is given explicitly as follows:
\ben
\phi^{Mott}(w)
& = & \frac{X^{Mott}}{\frac{d X^{Mott}(w)}{dw}}
= X\frac{d w^{Mott}}{dX} \\
& = & \frac{2}{(1-4X^2)^{3/2}}+\frac{1}{2X^2}\biggl(1- \frac{1}{\sqrt{1-4X^2}}\biggr)\\
& = & \sum_{n=0}^\infty (2n+1)\binom{2n+1}{n} X^{2n+1} \\
& = & 1+9X^2+50X^4+245X^6+1134X^8+5082X^{10}+\cdots.
\een
The coefficients  $1,9,50,245,1134,\dots$ are the sequence A001818 on OEIS \cite{Sloane}.
They are related to the Catalan numbers as follows:
\ben
(2n+1)\binom{2n+1}{n} = (2n+1)^2 \cdot \frac{1}{n+1}\binom{2n}{n}.
\een
The $\xi$-function in this case is:
\ben
\xi^{Mott}(w) & = & \int^w \frac{w}{\phi^{Mott}(w)}dw
= \int^w w \cdot \frac{\frac{dX}{dw}}{X} dX \\
& = & \int^w \frac{w}{X}dX = \int^u  \frac{1-\sqrt{1-4X^2}}{2X^2\sqrt{1-4X^2}}dX \\
& = & \frac{1-\sqrt{1-4X^2}}{2X} = \sum_{n=0}^\infty \frac{1}{n+1}\binom{2n}{n} X^{n+1},
\een
where the coefficients $\frac{1}{n+1}\binom{2n}{n}$ are the Catalan numbers.
To expand $\xi^{Mott}(w)$ to a power series in $w$,
we find:
\be \label{eqn:xi-Mott-in-w}
\xi^{Mott}(w)  =  w \sqrt{1-X(w)^2}.
\ee
Write $\sqrt{1-X^2}=Y$. Then by \eqref{eqn:w-Mott}
\ben
w^2=\frac{1-Y}{(1-Y^2)Y^2}.
\een
By Cardano formula we find
\ben
Y & = & \frac{1}{3w}(-w^3+3\sqrt{-3w^4+33w^2+3}+18w)^{1/3} \\
& + & \frac{w^2-3}{3w}(-w^3+3\sqrt{-3w^4+33w^2+3}+18w)^{-1/3}
-\frac{1}{3} \\
& = & 1-2w^2+10w^4-66w^6+498w^8-4066w^{10}+34970w^{12}-\cdots.
\een
The coefficients $1,2,10,66,498,\dots$ are
the sequence A027307 on OEIS \cite{Sloane}.
These are the 3-Schr\"oder numbers according to
Yang-Jiang \cite{Yang-Jiang},
i.e.,
the number of paths from $(0,0)$ to $(3n,0)$
that stay in first quadrant (but may touch horizontal axis)
and where each step is $(2,1)$, $(1,2)$ or $(1,-1)$.
Therefore,
\eqref{eqn:xi-Mott-in-w} and \eqref{eqn:w-Mott}
establish a relation between the Catalan numbers
and the 3-Schr\"oder numbers. See also Drake's Example 1.6.9 \cite{Drake}.

The entropy function in this case is
\ben
H^{Mott}(p)
& = & -p \int^p \frac{1}{w^2} \xi^{Mott}(w) dw \\
& = & -p \int^p \frac{1}{w^2} \frac{1-\sqrt{1-4X^2(w)}}{2X(w)} dw \\
& = & - p \int^{X(p)} \frac{1}{\biggl(\frac{1-\sqrt{1-4X^2}}{2X\sqrt{1-4X^2}}\biggr)^2}
\frac{1-\sqrt{1-4X^2}}{2X}
d \frac{1-\sqrt{1-4X^2}}{2X\sqrt{1-4X^2}} \\
& = & -p\log X(p) + p \sqrt{1-4X^2(p)}.
\een
One can directly check that
\ben
&& H^{Mott}(p) = F^{Mott}(X) - \log X \cdot X \frac{d F^{Mott}(X)}{d X}.
\een

\subsection{Conjugate sequence of dilogarithm function}

In this Subsection we present an example not included in Roman \cite{Roman}.
Consider the dilogarithm function:
\be
F(t) = \Li_2(t)= \sum_{n=1}^\infty \frac{t^n}{n^2}.
\ee
Although the series is only convergent for $|t| <1$,
but it can be rewritten as an integral:
\be
- \int_0^t \frac{\log (1-z)}{z}dz
= \int_0^t \frac{dz}{z} \int_0^z \frac{dw}{1-w},
\ee
one can use these integral to extend the definition of $\Li_2(t)$
to a multivalued function on $\bC-\{1\}$.

The conjugate sequence $\{\gamma_n(x)\}$ for the series $F(X)$ is defined by
\ben
\sum_{n=0}^\infty \gamma_n(x) \frac{t^n}{n!}
= \exp \biggl(x \sum_{n=1}^\infty \frac{t^n}{n^2}\biggr).
\een
The following are the first few terms of $\gamma_n$:
\ben
&& \gamma_1(x) = x, \\
&& \gamma_2(x) = x^2 + \frac{1}{2}x, \\
&& \gamma_3(x) = x^3 + \frac{3}{2}x^2+\frac{2}{3}x, \\
&& \gamma_4(x) = x^4 + 3x^3 +\frac{41}{12}x^2 + \frac{3}{2}x.
\een

To find the series $\phi$, we first find:
\be
w(X) = X\frac{d}{dX}F(X)= \sum_{n=1}^\infty \frac{X^n}{n} = \log \frac{1}{1-X},
\ee
then we apply Lagrange inversion to get:
\be
X = X(w) = 1- e^{-w}.
\ee
Next apply \eqref{eqn:phi-in-X} to get $\phi$:
\be
\phi(u) = \frac{1}{\frac{d}{du} \log X(u)}
= \frac{X(u)}{\frac{d}{du} X(u)} = e^u -1.
\ee
The $\phi$-logarithm is computed from $\phi(u)$ by the following formula:
\be
\ln_{\phi}(p)  =  \int^p_1 \frac{du}{\phi(u)} = \int_1^p \frac{u}{e^u-1}du
= \log \frac{e^p-1}{e^p} - \log \frac{e-1}{e},
\ee
this is again just $\log X(p)-\log X(1)$.
To get the $\phi$-entropy,
we first compute
\ben
\xi_{\phi}(p)
& = & \int_0^p \frac{u}{e^u-1}du
= p\log(1-e^{-p})-\sum_{n=1}^\infty \frac{e^{-np}}{n^2}
+ \sum_{n=1}^\infty \frac{1}{n^2} \\
& = & p \log (1-e^{-p}) - \Li_2(e^{-p}) + \frac{\pi^2}{6}.
\een
On the right-hand side of the last equality we see dilogarithm function again.
By Proposition \ref{prop:xi-F} we know that
\be
\xi_\phi(p) = F(X(p)) = \Li_2(1-e^{-p}).
\ee
They match each other by the following well-known identity due to Euler \cite{Lewin}:
\be
\Li_2(z) + \Li_2(1-z)=\frac{\pi^2}{6}-\log z\log(1-z).
\ee
Using
\ben
&& \frac{u}{e^u-1} = \sum_{n=0}^\infty \frac{B_n}{n!}u^n,
\een
one also has
\be \label{eqn:Bernoulli}
\xi_\phi(p) = \sum_{n=0}^\infty \frac{B_n}{(n+1)!}p^{n+1}.
\ee
By taking $p=1$ we get:
\be
\Li_2(1-e^{-1}) = \sum_{n=0}^\infty \frac{B_n}{(n+1)!}.
\ee
Next, the $\phi$-entropy can be computed as follows:
\ben
H_{\phi}(p) & = &  -p \int_1^p \frac{\xi_\phi(u)}{u^2}du
= -p \int_1^p \frac{\Li_2(1-e^{-u})}{u^2}du \\
& = & p \cdot \frac{\Li_2(1-e^{-u})}{u}\biggr|_1^p
- p \cdot \int_1^p \frac{1}{u} d \Li_2(1-e^{-u}) \\
& = &  \Li_2(1-e^{-p}) - p \cdot \Li_2(1-e^{-1})
- p \cdot \int_1^p \frac{1}{u}  \frac{udu}{e^u-1} \\
& = &  \Li_2(1-e^{-p}) - p \cdot \Li_2(1-e^{-1})
- p \cdot \log (1-e^{-u})|_1^p \\
& = & \Li_2(1-e^{-p}) - p \cdot \Li_2(1-e^{-1})
- p \cdot \log (1-e^{-p})
+ p \cdot \log (1-e^{-1}).
\een
This can be expanded as follows.
\ben
&& \frac{d}{dp} \log \frac{1-e^{-p}}{p}
= \frac{1}{e^p-1} - \frac{1}{p}
= \sum_{n=1}^\infty \frac{B_n}{n!} p^{n-1},
\een
and so
\ben
&& \log \frac{1-e^{-p}}{p}
= \sum_{n=1}^\infty \frac{B_n}{n!} \frac{p^n}{n}.
\een
By taking $p=1$ we get:
\ben
&& \log (1-e^{-1})
= \sum_{n=1}^\infty \frac{B_n}{n!} \frac{1}{n}.
\een
Hence we have
\ben
H_\phi(p) & = & -p \log p +
\sum_{n=0}^\infty \frac{B_n}{(n+1)!}p^{n+1}
- \sum_{n=1}^\infty \frac{B_n}{n!} \frac{p^{n+1}}{n} \\
& - & p \cdot [\Li_2(1-e^{-1})- \log (1-e^{-1})]  \\
& = &  -p \ln p - p \sum_{n=1}^\infty \frac{B_n}{n!}
\biggl(\frac{p^n}{n(n+1)}
- \frac{1}{n(n+1)}\biggr).
\een
On the other hand,
by \eqref{eqn:Bernoulli} we have:
\ben
H_{\phi}(p) & = &  -p \int_1^p \frac{\xi_\phi(u)}{u^2}du
= -p \ln p - p \sum_{n=1}^\infty \frac{B_n}{n!}
\biggl(\frac{p^n}{n(n+1)}
- \frac{1}{n(n+1)}\biggr) .
\een
It is a match.
In this case we have
\ben
H(X) & = & F(X) - w(X) \log X
= \Li_2(X) - \log X\cdot \log \frac{1}{1-X}.
\een
This is very close to Rogers dilogarithm:
\be
L(z) = \Li_2(z)+\half \log z \log (1-z).
\ee
Now
\ben
H(X(p)) & = & \Li_2(1-e^{-p})
- p \log (1-e^{-p}).
\een
So we get:
\ben
H(X(p)) = H_{\phi}(p) + p [\Li_2(1-e^{-1})- \log (1-e^{-1})].
\een
This matches with \eqref{eqn:H-F}.

\subsection{Averaged Acharya-Swamy statistics}

In this Subsection we present three more examples not included in Roman \cite{Roman}.
For the first example,
we take:
\ben
Y=F(X) = \frac{1}{2\epsilon} (\log (1+\epsilon X)
- \log (1- \epsilon X)),
\een
when $\epsilon=\frac{1}{2}$ this reduces to the case of Mittag-Leffler polynomials.
We have
\ben
X = f(Y) = \frac{e^{\epsilon Y}-e^{-\epsilon Y}}{e^{\epsilon Y}+e^{-\epsilon Y}},
\een
\ben
w(X) & = & X \frac{d}{dX} F(X)
= \frac{X}{2}\biggl(\frac{1}{1+\epsilon X}
+ \frac{1}{1-\epsilon X} \biggr)
= \frac{X}{1-\epsilon^2X^2}.
\een
From this we solve for $w$ to get:
\ben
X(w) & = & \frac{\sqrt{1+4\epsilon^2w^2}-1}{2\epsilon^2w}.
\een
It has the following series expansion:
\ben
X(w) = \sum_{n=0}^\infty (-1)^n\frac{1}{n+1}\binom{2n}{n} \epsilon^{2n}w^{2n+1},
\een
the coefficients $\frac{1}{n+1}\binom{2n}{n}$ are the Catalan numbers.
Next apply \eqref{eqn:phi-in-X} to get $\phi$:
\be
\phi(u) = \frac{1}{\frac{d}{du} \log X(u)}
= \frac{X(u)}{\frac{d}{du} X(u)}
= u\sqrt{1+4\epsilon^2u^2}.
\ee
It has the following expansion:
\ben
\phi(u) & = & u + 2u \sum_{n=0}^\infty (-1)^n \frac{1}{n+1}\binom{2n}{n} \epsilon^{2n+2}
u^{2n+2},
\een
where we again see the Catalan numbers.
The $\phi$-logarithm is computed from $\phi(u)$ by the following formula:
\ben
\log_{\phi}(p) & = & \int^p_1 \frac{du}{\phi(u)}
= \int_1^p \frac{1}{u\sqrt{1+4\epsilon^2u}}du \\
& = &  \log \frac{\sqrt{1+4\epsilon^2p^2}-1}{2\epsilon^2p}
-\log \frac{\sqrt{1+4\epsilon^2}-1}{2\epsilon^2},
\een
this is again just $\log X(p)-\log X(1)$.
It has the following expansion:
\ben
\log_{\phi}(p) & = & \log p
-\log \frac{\sqrt{1+4\epsilon^2}-1}{2\epsilon^2}
+ \sum_{n=1}^\infty (-1)^n \frac{1}{n}\binom{2n-1}{n} \epsilon^{2n}p^{2n},
\een
the coefficients $\binom{2n-1}{n}$ are the sequence A001700  on \cite{Sloane}.
Note
\ben
&& \frac{1}{n}\binom{2n-1}{n} = \frac{1}{2n} \binom{2n}{n},
\een
where $\binom{2n}{n}$ are the sequence A000984 on \cite{Sloane}.
To get the $\phi$-entropy,
we first compute
\ben
\xi_{\phi}(p)
& = &   \int^p_0 \frac{udu}{\phi(u)}
= \int_0^p \frac{u}{u\sqrt{1+4\epsilon^2u}}du
= \frac{1}{2\epsilon}\log (\sqrt{1+4\epsilon^2p^2}+2\epsilon p).
\een
One can check that it coincides with $F(X(p))$.
It has the following expansion:
\ben
\xi_{\phi}(p)
= \sum_{n=0}^\infty \frac{(-1)^n}{2n+1} \binom{2n}{n} \epsilon^{2n}p^{2n+1},
\een
the coefficients $\binom{2n}{n}$ are A000984 on \cite{Sloane}.
Next, the $\phi$-entropy can be computed as follows:
\ben
H_{\phi}(p) & = &  -p \int_1^p \frac{\xi_\phi(u)}{u^2}du
= -p \int_1^p \frac{1}{2\epsilon u^2}\log (\sqrt{1+4\epsilon^2u^2}+2\epsilon u)du .
\een
We can again check \eqref{eqn:H-F} in this case.
We have the following expansion:
\ben
H(X(p)) & = & F(X(p))-p \log X(p) \\
& = & \sum_{n=0}^\infty \frac{(-1)^n}{2n+1} \binom{2n}{n} \epsilon^{2n}p^{2n+1} \\
& - & p \log p - p \cdot \sum_{n=1}^\infty (-1)^n \frac{1}{2n}\binom{2n}{n} \epsilon^{2n}p^{2n} \\
& = & -p \log p
- \sum_{n=1}^\infty (-1)^n \frac{1}{2n(2n+1)}\binom{2n}{n} \epsilon^{2n}p^{2n+1} .
\een

Another way to take the average of Acharya-Swamy statistics is as follows:
\ben
Y=F(X) = \half \biggl(\frac{1}{\epsilon}\log (1+\epsilon X)
+ \epsilon \log (1+\frac{1}{\epsilon} X)\biggr),
\een

\ben
w(X) & = & X \frac{d}{dX} F(X)
= \frac{X}{2}\biggl(\frac{1}{1+\epsilon X}
+ \frac{1}{1+\frac{1}{\epsilon} X} \biggr) .
\een
From this we solve for $w$ to get:
\ben
X(w) & = &
\frac{1-(\epsilon+\epsilon^{-1})w
-\sqrt{1+(\epsilon-\epsilon^{-1})^2w^2}}{2w-(\epsilon+\epsilon^{-1})}.
\een
It has the following series expansion:
\ben
X(w) & = & w+sw^2+w^3+s(-s^2+2)w^4+(-s^2+2)w^5\\
& + & s(2s^4-6s^2+5)w^6+ (2s^4-6s^2+5)w^7 \\
& + & s(-5s^6+20s^4-28s^2+14)w^8 \\
& + & (-5s^6+20s^4-28s^2+14)w^9 \\
& + & (14s^9-70s^7+135s^5-120s^3+42s)w^{10} \\
& + & (14s^8-70s^6+135s^4-120s^2+42)w^{11} \\
& + & (-42s^{11}+252s^9-616s^7+770s^5-495s^3+132s)w^{12}
+\cdots,
\een
where $s$ is defined by
$$s= \frac{\epsilon+\epsilon^{-1}}{2}.$$
The coefficients appear as A094385, A094385, A157491  or A062991  on \cite{Sloane}.
Next apply \eqref{eqn:phi-in-X} to get $\phi$:
\ben
\phi(u) &= & \frac{1}{\frac{d}{du} \log X(u)} \\
& = & u\sqrt{1+4(s^2-1)u^2}
+ \frac{s}{2(s^2-1)}(\sqrt{1+4(s^2-1)u^2}-4(s^2-1)u^2-1).
\een
It has the following expansion:
\ben
\phi(u) & = & u + 2[(s^2-1)u^3-(s^2-1)^2u^5 +2(s^2-1)^3u^7 -5(s^2-1)^4u^9 \\
& +& 14(s^2-1)^5u^{11} + \cdots] \\
& + & s [-u^2-(s^2-1)u^4+2s(s^2-1)^2u^6-5(s^2-1)^3u^8\\
& + & +14(s^2-1)^4u^{10}
+\cdots],
\een
where the coefficients are the Catalan numbers.
The function $\frac{1}{\phi(u)}$
has the following expansion:
\ben
\frac{1}{\phi(u)}
& = & \frac{1}{u}+s+(-s^2+2)u+(-2s^3+3s)u^2+(3s^4-8s^2+6)u^3\\
& + & (6s^5-15s^3+10s)u^4+(-10s^6+36s^4-45s^2+20)u^5 \\
& + & (-20s^7+70s^5-84s^3+35s)u^6\\
&+&(35s^8-160s^6+280s^4-224s^2+70)u^7\\
&+&(70s^9-315s^7+540s^5-420s^3+126s)u^8\\
&+&(-126s^{10}+700s^8-1575s^6+1800s^4-1050s^2+252)u^9\\
&+&(-252s^{11}+1386s^9-3080s^7+3465s^5-1980s^3+462s)u^{10}
+ \cdots,
\een
to unravel the coefficients,
we note
\ben
\frac{s-u}{\phi(u)}
&= &\frac{s}{u}+s^2-1-s(s^2-1)u-2(s^2-1)^2u^2+3s(s^2-1)^2u^3\\
&+&6(s^2-1)^3u^4-10s(s^2-1)^3u^5-20(s^2-1)^4u^6\\
& + & 35s(s^2-1)^4u^7 + 70(s^2-1)^5u^8-126s(s^2-1)^5u^9 \\
&-&252(s^2-1)^6u^{10} +\cdots,
\een
the coefficients $1,1,2,3,6,10,20,35,70,126,\dots$ are the sequence A001405 on \cite{Sloane},
they are $\binom{n}{[n/2]}$ for $n=0,1,2, \dots$.
Using also
\begin{align*}
\xi_{\phi}(p)
& =    \int^p_0 \frac{udu}{\phi(u)}, &
H_{\phi}(p) & =   -p \int^p \frac{\xi_\phi(u)}{u^2}du,
\end{align*}
after integrations we get the expansion for $\log_{\phi}(p)$,
$\xi_\phi$ and $H_\phi$ from the above expansion for
$\frac{1}{\phi(u)}$.

Yet another way take the average of Acharya-Swamy statistics is as follows:
\ben
Y=F(X) = \half \biggl(\log (X+\epsilon )
+ \log (X+\frac{1}{\epsilon})\biggr),
\een
\ben
X = f(Y) = \frac{1}{2} \biggl(\epsilon +\frac{1}{\epsilon}
- \sqrt{\biggl(\epsilon -\frac{1}{\epsilon}\biggr)^2 + 4e^{2Y}}\biggr)
\een
\ben
w(X) & = & X \frac{d}{dX} F(X)
= \frac{X}{2}\biggl(\frac{1}{X+\epsilon}
+ \frac{1}{X+\frac{1}{\epsilon}} \biggr) .
\een
From this we solve for $w$ to get:
\ben
\frac{X(w)}{\epsilon+\epsilon^{-1}} & = &
\frac{(1-2w)
-\sqrt{1
-4\biggl(\frac{\epsilon-\epsilon^{-1}}
{\epsilon+\epsilon^{-1}}\biggr)^2(w-w^2)}}{4(w-1)}.
\een
It has the following series expansion:
\ben
\frac{2X(w)}{\epsilon+\epsilon^{-1} }
& = & (1-t)w+(1-t^2)w^2+(1+t^2-2t^3)w^3 \\
& +& (1+ 4t^3-5t^4)w^4+(1-2t^3+15t^4-14t^5)w^5 \\
& +& (1-15t^4+56t^5-42t^6)w^6 \\
&+&(1+5t^4-84t^5+210t^6-132t^7)w^7\\
&+&(1-420t^6+56t^5+792t^7-429t^8)w^8 \\
&+&(1-14t^5+420t^6-1980t^7 +3003t^8-1430t^9)w^9
+\cdots,
\een
where $t$ is defined by
$$t= \biggl(\frac{\epsilon-\epsilon^{-1}}
{\epsilon+\epsilon^{-1}}\biggr)^2,$$
the coefficients $14,420,1980,3003,1430$ etc.
are the sequence A117434 on \cite{Sloane}.
Next apply \eqref{eqn:phi-in-X} to get $\phi$:
\ben
\phi(u) = \frac{1}{\frac{d}{du} \log X(u)}
= \frac{-1}{2t}(1-4t(u-u^2)-\sqrt{1-4t(u-u^2)}).
\een
It has the following expansion:
\ben
\phi(u) & = & u-(t+1)u^2-(2t^2-2t )u^3
-(5t^3-6t^2+t) u^4 \\
& -& (14t^4-20t^3+6t^2)u^5
-(42t^6-70t^5+30t^3-2t^2)u^6 \\
& - & (132t^6-252t^5+140t^4-20t^3)u^7 \\
& - & (429t^7-924t^6+630t^5-140t^4+5t^3)u^8 \\
& - & (1430t^8-3432t^7+2772t^6-840t^5+70t^4)u^9
+\cdots,
\een
where the coefficients $1430,3432,2772,840,70$ etc
are the sequence A068763.
In this case $X(1) = \infty$
we weil define the $\phi$-logarithm   $\log_{\phi}(p)$
to be just $\log X(p)$.
The function $\frac{1}{\phi(u)}$
has the following expansion:
\ben
\frac{1}{\phi(u)}
& = & \frac{1}{u}+(t+1)+(3t^2+1)u+(10t^3-3t^2+1)u^2 \\
& + & (35t^4-20t^3+1)u^3+(126t^5-105t^4+10t^3+1)u^4\\
& + & (462t^6-504t^5+105t^4+1)u^5 \\
& + & (1716t^7-2310t^6+756t^5-35t^4+1)u^6\\
& + & (6435t^8-10296t^7+4620t^6-504t^5+1)u^7\\
& + & (24310t^9-45045t^8+25740t^7-4620t^6+126t^5+1)u^8
+ \cdots,
\een
the coefficients $1,3,10,126,\dots$ are the sequence A001700  on \cite{Sloane},
they are $\binom{2n-1}{n}$ for $n=1,2, \dots$.
The coefficients $3,20,105,504,2310,\dots$ are the sequence
A000917 on \cite{Sloane}. They are
$(2n-1)!/(n!(n-2)!)$ for $n=2,3,4,\dots$.
The sequence of coefficients $10,105,756,4620,25740,\dots$
are not on \cite{Sloane}.
However we note
\ben
\frac{1}{\phi(u)}
& = & \frac{1-4t(u-u^2)+\sqrt{(1-4t(u-u^2)}}{2u(1-u)(1-4t(u-u^2))},
\een
hence
\ben
\frac{2u(1-u)}{\phi(u)}
& = &  1 +  \frac{1}{\sqrt{1-2 \cdot t^{1/2} \cdot (2t^{1/2}u)  + (2t^{1/2} u)^2}}\\
& = &
2+t\cdot (2u) +\frac{1}{2}(-t+3t^2)(2u)^2
+\frac{1}{2}(-3t^2+5t^3)(2u)^3 \\
& + & \frac{1}{8}(3t^2-30t^3+35t^4)(2u)^4
+ \frac{1}{8}(15t^3-70t^4+63t^5)(2u)^5 \\
& + & \frac{1}{16}(-5t^3+105t^4-315t^5+231t^6)(2u)^6\\
&+& \frac{1}{16}(-35t^4+315t^5-693t^6+429t^7)(2u)^7 \\
& +& \frac{1}{128}(35t^4-1260t^5+6930t^6-12012t^7+6435t^8)(2u)^8 \\
& +& \frac{1}{128}(-9240t^6+630t^5+36036t^7-51480t^8+24310t^9)(2u)^9+\cdots,
\een
which is essentially the generating function of the Legendre polynomials.
Again after integrations
we get the expansion for $\log_{\phi}(p)$,
$\xi_\phi$ and $H_\phi$ from the above expansion for
$\frac{1}{\phi(u)}$.

\subsection{Bell polynomials and universal interpolating statistics}

The incomplete Bell polynomials are defined as follows:
\be
\begin{split}
& B_{n,k}(t_1, \dots, t_{n-k+1}) \\
= & \sum \frac{n!}{m_1!\cdots m_{n-k+1}!} \biggl(\frac{t_1}{1!}\biggr)^{m_1}
\cdots \biggl(\frac{t_{n-k+1}}{(n-k+1)!}\biggr)^{m_{n-k+1}},
\end{split}
\ee
where the sum is taken over all sequences $m_1, m_2, m_3, \dots, m_{n-k+1}$
of non-negative integers such that these two conditions are satisfied:
\ben
&& m_1 + m_2   + \cdots + m_{n - k + 1}   = k ,  \\
&& m_{1}+2m_{2}+\cdots +(n-k+1)m_{n-k+1}=n.
\een
From this definition it is easy to see that when
\be \label{eqn:F-Bell}
F^{Bell}(X) = \sum_{j=1}^\infty t_j \frac{X^j}{j!},
\ee
one has
\be
\exp (xF^{Bell}(X)) = 1 + \sum_{n=1}^\infty \frac{X^n}{n!}
\sum_{k=1}^n B_{n,k}(t_1, \dots, t_{n-k+1}) x^k.
\ee
Therefore,
if one takes $t_1=1$,
then
\be
\gamma^{(n)}(x; \bt) = \sum_{k=1}^n B_{n,k}(t_1, \dots, t_{n-k+1}) x^k
\ee
defines a polynomial sequence of binomial type.
The spectral curves in this case are
\begin{align}
z &= \exp (X+ \sum_{j > 1} \frac{t_j}{j!} X^j), &
Y& = X+ \sum_{j > 1} \frac{t_j}{j!} X^j.
\end{align}
These give  the indication that
the Bell polynomials give the universal interpolating statistics
and the spectral curves can be understood as
the universal spectral curve.
All other examples are given by choosing $t_j$'s in a suitable way.

In this case,
the Lagrange inversion formula that finds the coefficient $s_j$ in the inverse series
\be
f^{Bell}(X) = \sum_{j=1}^\infty s_j \frac{X^j}{j!}
\ee
from \eqref{eqn:F-Bell} has different combinatorial interpretations involving
summations over trees
and has appeared in the literature in many times.
For example,
in Chen's thesis \cite{Che} in 1990,
it is interpreted as a summation over all Schr\"oder trees on $n$ vertices.
In his 2008 thesis,
Drake \cite[Example 1.4.7]{Drake} interpreted it as a summation over phylogenetic trees.
This was rediscovered in a 2016 paper by Engbers-Galvin-Smyth \cite[Lemma 4.2]{Eng-Gal-Smy}.
This is a special case of the  general formulas for Lagrange inversion in terms of
summation over trees published by Wright \cite{Wright} in 1989.

Wright's formula was rediscovered by Ginzburg and Kapranov \cite[Theorem 3.3.9]{Gin-Kap}
in the setting of Kozsul duality of operads.
This suggests a possible connection of this work with the theory of operads.
Indeed,
Gingzburg and Kapranov defined  the generating series of a $k$-linear operad
$\cP(n)$ by
\be
g_{\cP}(x) = \sum_{n=1}^\infty \dim \cP(n)\cdot \frac{x^n}{n!}.
\ee
Denote by $\cA s$, $\cC om$, $\cL ie$ the associative operad, the commutative operad,
and the Lie operad£¬ respectively.
Then by \cite[Example (3.1.12)(a)]{Gin-Kap}:
\begin{align*}
g_{\cA s}(x) & = \frac{x}{1-x}, & g_{\cC om}(x)& = e^x -1, &
g_{\cL ie}(x) & = - \log (1-x).
\end{align*}
We recognize that $g_{\cA s}(X) = w^{BE}(X)$ is the weight function for the Bose-Einstein
statistics,
it is also the free energy of the example in \S \ref{sec:Gentile-Mot};
$g_{\cC om}(x) = e^x -1$ as the free energy associated with the exponential
polynomials (cf. \S \ref{sec:Exponential}); $g_{\cL ie}(x)$ as
the free energy function of the Bose-Einstein statistics.
By \cite[Theorem 3.3.9]{Gin-Kap},
if $\cP$ is an admissible dg-operad and $\cQ$ is its dual,
then
\be
g_{\cQ}(-g_{\cP}(-x)) = x.
\ee
Note $\cA s$ is self-dual, and $\cC om$ and $\cL ie$ are dual to each other.

\end{appendix}

\end{document}